\documentclass{lmcs} 


\usepackage{amsfonts}
\usepackage{amsmath}
\usepackage{graphicx}
\usepackage{float} 

\usepackage{url}
\usepackage{amssymb}

\usepackage{subcaption}

\usepackage{hyperref}

\usepackage{tikz}
\usetikzlibrary{positioning}
\definecolor{graphblue}{RGB}{91, 132, 177}
\definecolor{heatblue}{RGB}{43, 85, 140}       
\definecolor{heatgrey}{RGB}{125, 117, 129}       
\definecolor{heatdarkorange}{RGB}{227, 154, 71} 
\definecolor{heatlightorange}{RGB}{246, 198, 125}
\definecolor{heatcream}{RGB}{254, 241, 203}      


\begin{document}\bibliographystyle{alphaurl}		
\title[Probabilistic Disjunctive Normal Forms]{Probabilistic Disjunctive Normal Forms in Temporal Logic and Automata Theory}

\author[A.~Kuznetsov]{Alexander Kuznetsov\lmcsorcid{0000-0002-0365-2077}}[a,b]

\address{
Indian Institute of Technology Jammu, Jagti, NH-44, PO Nagrota Jammu -- 181 221 J\&K, India}
\address{V.A. Trapeznikov Institute of Control Sciences of Russian Academy of Sciences,	Profsoyuznaya st., 65, Moscow, 117997,	Russia}	
\email{avkuz@bk.ru}  
	
		\begin{abstract}					
			\noindent This article introduces probabilistic disjunctive normal forms (PDNFs) as a framework for representing and reasoning about uncertainty in logical systems. Unlike classical DNFs, PDNFs assign real-valued weights to variables, encoding probabilistic information about their presence, absence, or negation. Then we construct a vector space of PDNFs that allows algebraic evidence combination. PDNFs are interpreted as probability distributions over venjunctions (temporal logic constructs) and as integrable functions over partitioned intervals, where the integrals determine variable probabilities. This dual perspective allows for a Banach space structure and the application of functional analysis. We demonstrate that, under exponential parametrisation, PDNF addition aligns with Bayesian evidence fusion and derive bounds for outcome identification from random samples. The formalism thus bridges logic, numerical methods, and continuous probability.
		\end{abstract}

\maketitle

\section{Introduction}

Consider a system with $m$ sensors, observed over a sequence of time steps. Each sensor may or may not be triggered at a given moment, and several parameters can influence its response -- both those initially intended by the user (such as motion, mass, illumination, humidity) and random factors. The presence or absence of specific random factors is unknown to the user, but they know exactly how and which factors could potentially influence the device's operation. It is necessary to develop a mathematical formalism that would describe the operation of such a sensor system in the absence of knowledge about the effects of random factors, but with knowledge of how the sensor will behave in response to each factor. Essentially, we should implement a mechanism for complex sensor estimation similar to \cite{SERGEEV2021134}. Alternatively, consider a standard problem in neurosymbolic Artificial Intelligence: we need to synthesise a set of statements that are syntactically or even semantically close to a given target formula, or conversely, to evaluate whether a collection of observed statements cluster around a specific ``central'' belief. While modern large language models operate implicitly as black-box sequence generators over discrete vocabulary tokens via continuous latent representations, we want a formalism that offers an explicit framework.

The formalism developed in this article can be connected with the broader family of models for reasoning under uncertainty, for example, with Probabilistic Boolean Networks, Markov Logic Networks, and Dempster-Shafer Theory. 

Probabilistic Boolean Networks \cite{shmulevich2002probabilistic,Drossel2008} model dynamical systems by assigning a probability distribution over sets of deterministic Boolean update functions for each variable. In contrast, our model associates a real-valued weight with each variable at each step to govern its stochastic presence, negation, or absence within a logical conjunction, endowing the space of formulae with a linear structure.

Markov Logic Networks \cite{richardson2006markov} assign real-valued weights to first-order logic formulae as soft constraints, defining a joint distribution over possible worlds proportional to the exponential of the weighted sum of satisfied formula groundings. Our method, by contrast, assigns weights directly to individual variables within ordered conjunctive clauses, yielding a representation equipped with a linear vector space structure.

Dempster--Shafer Theory \cite{shafer1976mathematical} assigns basic probability masses to subsets of hypotheses in order to derive belief and plausibility measures, providing a general framework for non-additive evidence combination at the cost of exponential complexity over power sets. By focusing on variable-wise probabilities rather than arbitrary subset masses, our formalism trades some of this representational generality for a tractable linear vector space structure, directly encoding uncertainty into the variables of a logical expression via parametric functions.

On the other hand, our design is reminiscent of fuzzy logic \cite{zadeh1965fuzzy} (as well as Nilsson's probabilistic logic \cite{NILSSON198671}), but with a fundamental semantic difference: instead of assigning a degree of truth to a proposition, we model the probability of a literal's syntactic occurrence, negation, or omission within a formula, derived from the physical characteristics of the sensor.

In our model, the first level of uncertainty is governed by the mapping $\mathbf{F}$, which for a fixed weight matrix $\Xi$ determines the occurrence probabilities of individual literals (analogous to primary membership degrees in type-1 fuzzy sets). The second level arises when the weights $\Xi$ are themselves random, making the induced probability distribution over temporal structures (venjunctions) a random measure. This two-tiered structure is conceptually analogous to type-2 fuzzy sets, where primary membership values are themselves uncertain and governed by a secondary membership function \cite{KarnikMendel2001}. In this analogy, taking the Bochner integral of a random encoder or random measure plays a role comparable to type-reduction in type-2 fuzzy systems, averaging out the second-level uncertainty.

A related algebraic-logical perspective is presented in \cite{sitnikov2002method}, which addresses logical correlation in knowledge bases by representing dependencies among data features as systems of logical equations over finite predicates. Whereas that approach operates with deterministic finite-predicate equations, our formalism introduces real-valued probabilistic weights on literals and embeds the resulting sequential expressions into a normed functional space.

From another perspective, the proposed construction can be understood as defining, using a probabilistic formula, a family of conventional formulae that are similar in some sense. We select the formula that best describes the object of interest to us. This corresponds, for example, to developing a proof for a theorem similar to one already known. In this sense, this work continues the article \cite{Schumann04052018}.

Finally, the proposed formalism, based on asynchronous sequential logic \cite{vasyukevich2011digital} and used to express memory effects, relies on moments of a system’s state change. The theoretical framework of sequential logic consists of mathematical tools such as sequention and venjunction, as well as the logical-algebraic equations built upon them.

The recent work on parametric probabilistic automata (pPA) introduces parameters into transition probabilities, enabling compositional reasoning about system properties \cite{mertens2025parametric}. In parallel, research on knowledge compilation, such as the work on deterministic decomposable negation normal form (d-DNNF) circuits, seeks efficient representations for probabilistic inference by exploiting the structure of Boolean circuits \cite{derkinderen2025compiling}. Meanwhile, extensions of temporal logic, such as PTATL, enable the verification of strategic properties in multi-agent systems under uncertainty \cite{kwiatkowska2025probabilistic}.

While these approaches focus on compositional verification of system dynamics (pPA), tractable inference via circuit compilation (d-DNNF), or strategic model checking in multi-agent systems (PTATL), our work embeds uncertain sequential formulae into a normed vector space and represents them as integrable functions in Banach spaces.

All tasks achievable with the proposed formalism can also be formulated using probabilistic regular expressions or Bayesian networks, but the primary motivation for the formalism is to provide a framework that closely resembles ordinary Boolean logic while also carrying the structure of a linear space. In our framework, formulae look like familiar logical expressions; semantically, however, a probabilistic elementary conjunction and a PDNF should be understood as a sequential log of incoming evidence weighted by degrees of confidence, like lineage and provenance representations in probabilistic databases \cite{green2007provenance,suciu2011probabilistic}, rather than as static truth-functional propositions in classical logic. At the same time, we can perform operations such as addition, scalar multiplication, and use norm-based similarity measures. In this paper, we discuss probabilistic disjunctive normal forms (the formalism, previously constructed by the author in \cite{Kuznetsov2025PDNF}) and investigate their fundamental properties. We obtain quantitative estimates of the number of observations needed to completely determine the underlying deterministic structure from random realisations of a probabilistic disjunctive normal form. Furthermore, we prove several theorems on the identifiability of such a deterministic description, showing how repeated observations gradually eliminate uncertainty and eventually allow us to replace the probabilistic model by an explicit list of possible behaviours.   

The need for a continuous functional-analytic framework of the present paper arises from two distinct sources of continuity in reasoning under uncertainty. First, even when a probabilistic disjunctive normal form has a finite number of variables, the underlying weight matrix $\Xi$ may be governed by an absolutely continuous probability distribution driven by physical sensor noise or stochastic dynamics. Second, the logical expressions themselves may become infinitary when propositions concern domain entities with continuous attributes (e.g., physical measurements or semantic dimensions) evolving over continuous time. As demonstrated in the article on knowledge graph metrics-based formal argumentation \cite{KuznetsovLemtyuzhnikova2026} (or in its abridged English version \cite{kuznetsov2026modelling}), while discrete logical truth is inherently discontinuous and breaks down under infinitesimal perturbations, an agent's epistemic neighbourhood (a ``cloud'' of plausible belief graphs) varies smoothly over a metric space. Encoding statements as integrable functions in Banach spaces (e.g., $L_1$, $L_2$, or $L_\infty$) translates continuous variations in object attributes or weight distributions into smooth transformations of induced probability measures.
  
While these results have counterparts in other probabilistic frameworks, the main contribution of this paper is the systematic connection between discrete logical structures and continuous functional analysis, which opens new possibilities for applying analytic methods to logical problems. Basically, the main idea of the present article is to encode a probabilistic logical expression via an integrable function.

\section{Probabilistic Disjunctive Normal Forms and their Properties}

In the article \cite{Kuznetsov2025PDNF}, the author introduced probabilistic disjunctive normal form (PDNF) as a tool to describe a system with uncertainty. Let us recall the definitions from the aforementioned article, with some minor additions, and then (Subsection~\ref{ss:cont}) move from a discrete to a continuous model of logical functions.
 
\begin{defi}\label{def:def01}
	We call the probabilistic elementary conjunction of $m$ variables $x_1,\dots , x_m$ the form $x_1^{\xi_1}\wedge \dots \wedge x_m^{\xi_m}$, where $\xi_j\in \mathbb{R}$, $j=\overline{1,m}$. Assume that $F_1^j:\mathbb{R}\rightarrow [0,1]$ is a non-decreasing mapping, $F_{-1}^j:\mathbb{R}\rightarrow [0,1]$ is a non-increasing mapping, $F_0^j:\mathbb{R}\rightarrow [0,1]$, $F_\ell = (F_\ell^1,\dots, F_\ell^m)$, and $F^j_{-1}(\xi)+F^j_{0}(\xi)+F^j_{1}(\xi) = 1$ for all $\xi\in\mathbb{R}$, $F_{\ell}^j$ are continuous except perhaps at a finite number of points, $\ell=\overline{-1,1}$, $j=\overline{1,m}$. We define that 
	\begin{gather}
	\mathbb{P}\{x_j^{\xi_j}\sim x_j\} =F_{1}^j(\xi_j),\label{eqPF1}\\
	\mathbb{P}\{x_j^{\xi_j}\sim \varepsilon\}= F_0^j(\xi_j),\\
	\mathbb{P}\{x_j^{\xi_j}\sim  \overline{x}_j \} = F_{-1}^j(\xi_j),\label{eqPF3}
	\end{gather}
	Here ``$x_j^{\xi_j}\sim x_j$'' means that the variable $x_j$ is present in the realisation of the probabilistic elementary conjunction in the place of $x_j^{\xi_j}$, and $x_j^{\xi_j}\sim \varepsilon$ means that the  variable $x_j$ is not present in the realisation of the probabilistic elementary conjunction in the place of $x_j^{\xi_j}$. 
\end{defi}

\begin{defi}
	We call the PDNF of $m$ variables with the weight mapping $\mathbf{F}:\mathbb{R}^m\rightarrow ([0,1]^{m})^3$ the form $Z=X_1\sqcup\dots \sqcup X_n$, where $X_i$, $i=\overline{1,n}$ are probabilistic elementary conjunctions of $m$ variables, and $\mathbf{F}=(F_{-1}, F_{0}, F_1)$ consists of the functions from the previous definition. Let us denote the set of such PDNFs as $\mathcal{B}^m_{\mathbf{F}}(n)$. 
\end{defi}

\begin{rem}\label{rem:zero_weight}
In Sections~\ref{sec:vasuk}, \ref{sec:newsec}, and~\ref{sec:Markov}, where $\mathbf{F}$ describes a known physical or logical mechanism generating realisations, we impose the condition $F_1^j(0) = F_{-1}^j(0) = 0$ (and hence $F_0^j(0) = 1$, which by monotonicity implies $F_1^j(\xi) = 0$ for $\xi < 0$ and $F_{-1}^j(\xi) = 0$ for $\xi > 0$). Under this convention, if $\xi_j = 0$, then $x_j^{\xi_j}$ surely does not appear in the elementary conjunction, allowing us to assume without loss of generality that all PDNFs have equal length $n$ by padding shorter forms with empty probabilistic elementary conjunctions $x_1^0\wedge \dots \wedge x_m^0$.
\end{rem}

The PDNF is a syntactic construction; the square disjunction sign ``$\sqcup$'' is merely a separator between ordered conjunctions. Possible semantic interpretations will be described in Sections~\ref{sec:vasuk},~\ref{sec:newsec}, and \ref{sec:Markov}.

\begin{exa}[A Soviet Story]\label{exa:ss}
Let us illustrate how a PDNF can serve as a compact record of incoming, uncertain reports about a sequence of events. Consider a causal narrative: ``Vlad wrote a denunciation, and I was fired''.  
In a fully deterministic setting, this could be represented as a formula  
\[
(V \wedge \neg I \wedge \neg P)\;\sqcup\; (W \wedge \neg O) \;\sqcup\;(F \wedge \neg J \wedge \neg D),
\]
where the variables have the following meanings:
\begin{itemize}[label=$\triangleright$]
	\item $V$: ``Vlad creates a denunciation'', $I$: ``Ivan creates a denunciation'', $P$: ``Peter creates a denunciation'';
	\item $W$: ``written denunciation'', $O$: ``oral report'';
	\item $F$: ``I was fired'', $J$: ``I was imprisoned'', $D$: ``I was demoted''.
\end{itemize}
The $\sqcup$ operator indicates that the groups of events are observed in this chronological order. In reality, however, incoming reports are incomplete and uncertain. The events within each time step are not assumed to be mutually exclusive: multiple colleagues may independently file denunciations, a report may be delivered both orally and in writing, and several administrative consequences may occur simultaneously. Suppose that, based on prior data, the independent marginal probabilities of these events being recorded at their respective steps are:
\begin{itemize}[label=$\triangleright$]
	\item at step $1$, reports of denunciations by Vlad, Ivan, and Peter appear independently with probabilities $1/2$, $1/3$, and $1/6$, respectively;
	\item at step $2$, a written denunciation and an oral report are recorded independently with probabilities $1/3$ and $2/3$;
	\item at step $3$, firing, imprisonment, and demotion occur independently with probability $1/3$ each.
\end{itemize}

Let the weight mapping $\mathbf{F}$ be defined component-wise by $F^i_1(\xi) = \min(\max(\xi, 0), 1)$, $F^i_{-1}(\xi) = \min(\max(-\xi, 0), 1)$, and $F^i_0(\xi) = 1 - F^i_1(\xi) - F^i_{-1}(\xi)$ for each variable $i$ in the ordered list $(V,I,P,W,O,F,J,D)$, so that $F^i_1(\xi_i) = \xi_i$ for $\xi_i \in [0,1]$ and $F^i_0(0)=1$. Omitting literals with zero weights ($\xi_i = 0$) for readability, the entire uncertain dossier is represented by the PDNF
\[
(V^{1/2} \wedge I^{1/3} \wedge P^{1/6})\;\sqcup\; (W^{1/3} \wedge O^{2/3}) \;\sqcup\;(F^{1/3} \wedge J^{1/3} \wedge D^{1/3}).
\]
Under the coordinate-wise independence of PDNF realisations, this expression induces a product distribution over all possible event logs, allowing us, for example, to compute the likelihood of a specific combined trajectory such as ``Ivan filed an oral (and not a written) report without Vlad or Peter participating, and I was imprisoned, but neither fired nor demoted'' (which has probability $(1 - 1/2)\cdot(1/3)\cdot(1 - 1/6) \cdot (1 - 1/3)\cdot(2/3) \cdot (1 - 1/3)\cdot(1/3)\cdot(1 - 1/3) = 20/2187$), or to test which variants remain consistent with new evidence. The weights $\xi_i$ in $x_i^{\xi_i}$ can also be interpreted as the agent's subjective confidence or degree of belief regarding the presence, absence, or negation of the respective $x_i$.
\end{exa}

\begin{defi}
	We call the product of the probabilistic elementary conjunction $X=x_1^{\xi_1}\wedge \dots \wedge x_m^{\xi_m}$ and the number $\alpha\in \mathbb{R}$ the probabilistic elementary conjunction
	\[\alpha \cdot X = x_1^{\alpha \xi_1}\wedge \dots \wedge x_m^{\alpha \xi_m}.\]
	
	We call the product of the probabilistic DNF $Z = \bigsqcup_{i=1}^n X_i$ and the number $\alpha\in \mathbb{R}$ the form
	\begin{equation}\label{eq:mult}
	\alpha Z=\alpha \bigsqcup_{i=1}^n X_i = \bigsqcup_{i=1}^n (\alpha X_i).
    \end{equation}
\end{defi}

\begin{defi}
	We call the sum of probabilistic elementary conjunctions $X=x_1^{\xi_1}\wedge \dots \wedge x_m^{\xi_m}$ and $Y=x_1^{\eta_1}\wedge \dots \wedge x_m^{\eta_m}$ the probabilistic elementary conjunction
	\[X+ Y = x_1^{\xi_1+\eta_1}\wedge \dots \wedge x_m^{\xi_m+\eta_m}.\]
	We call the sum of probabilistic DNF $Z_1 = \bigsqcup_{i=1}^n X_i$, $Z_2 = \bigsqcup_{i=1}^n Y_i$ the following form
	\begin{equation}
	Z_1+ Z_2=\bigsqcup_{i=1}^n X_i+ \bigsqcup_{i=1}^n Y_i = \bigsqcup_{i=1}^n (X_i+ Y_i).\label{eqsum}
	\end{equation}
\end{defi}

The set $\mathcal{B}^m_\mathbf{F}(n)$ with the aforementioned operations forms a linear vector space. The set of all probabilistic elementary conjunctions of $m$ variables can be represented as $\mathbb{R}^m$, the set of all probabilistic DNF of $m$ variables and $n$ conjunctions as the set of matrices $\mathbb{R}^{n\times m}$.

The space $\mathcal{B}^m_\mathbf{F}(n)$ can be normed using the norm induced from $\mathbb{R}^m$ as follows:
\begin{equation}\label{eq:norm1}
\|X_i\|_1=\| x_1^{\xi_{i1}}\wedge \dots \wedge x_m^{\xi_{im}}\|_1=\sum_{j=1}^m |\xi_{ij}|,\quad
\biggl \|\bigsqcup_{i=1}^n X_i\biggr \|_1= \sum_{i=1}^n \|X_i\|_1.
\end{equation}

Obviously, as the triangle inequality
\[\|Z_1+ Z_2\|_1=\sum_{i=1}^n\sum_{j=1}^m |\xi_{ij}+\eta_{ij}|\leq \|Z_1\|_1+\|Z_2\|_1\]
holds, the $\|\cdot\|_1$ is a sound norm on the vector space.

The idea behind introducing an algebraic structure for PDNFs is to model the aggregation of evidence from multiple observational series (see Appendix \ref{app:algebra}). Additional explanation of this kind of summation will be given in Section \ref{sec:automata}. Additional ideas about the internal structure of PDNF are contained in Appendix \ref{app:struct}.

\subsection{Transition from discrete structure to continuous functions\label{ss:cont}}

Let us denote $\Delta_{tj}=[j-1+(t-1)m,j+(t-1)m)$ and define a Riemann integrable function $X(t,\xi_{tj};\cdot):\Delta_{tj}\rightarrow \mathbb{R}$ with support $\Delta_{tj}$, $t\geq 1$, $\xi_t = (\xi_{t1},\dots, \xi_{tm})\in \mathbb{R}^m$. Setting $I_{tj} = \int_{\Delta_{tj}} X(t,\xi_{tj};x)\,dx$, we define the literal probabilities as
\begin{gather}
	\mathbb{P}\{x_j^{\xi_{tj}}\sim x_j\} = \min\bigl(\max(I_{tj}, 0), 1\bigr),\label{eq:cont01}\\
	\mathbb{P}\{x_j^{\xi_{tj}}\sim \overline{x}_j\} = \min\bigl(\max(-I_{tj}, 0), 1\bigr),\\
	\mathbb{P}\{x_j^{\xi_{tj}}\sim \varepsilon\} = 1 - \mathbb{P}\{x_j^{\xi_{tj}}\sim x_j\} - \mathbb{P}\{x_j^{\xi_{tj}}\sim \overline{x}_j\}.\label{eq:cont03}
\end{gather}

Here $\xi_{tj}$ in $X(t,\xi_{tj};\cdot)$ is a parameter defining the shape of the function $X(t,\xi_{tj};\cdot)$. If, for example, the function $X(t,\xi_{tj};\cdot) \equiv\alpha_{tj}$ is a constant on $\Delta_{tj}$, we may set $\xi_{tj} = \alpha_{tj}$.

Let us define an encoder of a probabilistic elementary conjunction at step $t$ as
\[
Y(t,\xi_t;x) = \sum_{j=1}^m X(t,\xi_{tj};x).
\]

Let $\xi:[0,n]\to \mathbb{R}^m$, $\xi: t\mapsto \xi_t$, $t=\overline{1,n}$, be a function with certain properties (e.g., a continuous one, a random walk, etc.). A PDNF corresponds to the sum
\[
z(\xi; x) = \sum_{t=1}^n Y(t,\xi(t); x).
\]
If, for example, the functions $X(t,\xi_{tj};\cdot)$ are constants on $\Delta_{tj}$, then the corresponding encoder 
\[
z(\xi; x)=\sum_{t=1}^n\sum_{j=1}^m\xi_{tj}\mathbf{1}_{\Delta_{tj}}(x),\quad \xi_j(t) = \xi_{tj},
\]
is a piecewise constant function and $\xi_{tj}$ is the height (negative or positive) of the step on $\Delta_{tj}$.

It is clear that a classical DNF in this notation can be easily encoded as a piecewise constant function in which possible values of $\xi_{tj}$ belong to the set $\{-1,0,1\}$. The function $z(\xi; \cdot)$ can be seen as an analogue signal encoding the corresponding logical statement with a degree of reliability defined by $\xi$. Thus, the sum and scalar multiplication of PDNFs defined by (\ref{eqsum}), (\ref{eq:mult}) can be understood in terms of signal addition and amplification.

Therefore, any piecewise continuous function $z(\xi;\cdot): [0,nm]\rightarrow \mathbb{R}$ with a finite number of jumps encodes a PDNF from $\mathcal{B}^m_{\mathbf{F}}(n)$. We denote the set of these encoders as $\mathcal{E}^m_{\mathbf{F}}(n)$. For $\mathcal{E}^m_{\mathbf{F}}(n)$, the $L_1$-norm is given by
\[\|z(\xi; x)\|_1 = \int_0^{nm}|z(\xi; x)|dx.\]
Note that the discrete norm (\ref{eq:norm1}) is defined directly on the parameter space of weights $\mathbb{R}^{n \times m}$, and coincides with $\|z(\xi; \cdot)\|_{L_1}$ when $z$ is piecewise constant on the unit-length intervals $\Delta_{tj}$. We will further suppose that the set $\mathcal{E}^m_{\mathbf{F}}(n)$ is bounded in the sense of $\|\cdot\|_1$ or even $\|\cdot\|_\infty$ supremum norm when necessary. 

\begin{rem}
The aforementioned way of encoding was chosen only for simplicity. If we want to use continuous time, we should consider encoders $z(\xi;\cdot,\cdot):[0,n]\times [0,m]\rightarrow \mathbb{R}$, where the set $\Delta_{tj}=[t-1,t)\times[j-1,j)$ corresponds to the $j-$th variable in the $t-$th probabilistic elementary conjunction. In this case, we can model PDNFs as a continuum of elementary conjunctions, interpreted as a continuous evolution of the agent’s state of knowledge.

The probabilities (\ref{eq:cont01})--(\ref{eq:cont03}) for this type of encoding can be expressed via the signed integral
\[I_{tj}=\iint_{[t-1,t)\times[j-1,j)}z(\xi;\tau,x)d\tau dx,\quad t=\overline{1,n}, j=\overline{1,m}.\]
We will denote the set of encoders of this type as $\mathcal{E}_{\mathbf{F}}(n\times m)$.
\end{rem}

\begin{rem}
	We can also use encoders to encode the probabilities $x_{ij} = \mathrm{Truth}$, $x_{ij} = \mathrm{False}$ or that $x_{ij}$ is undefined. To handle more than ternary-valued logic, we should use, for example, $z(\xi;\cdot,\cdot):[0,n]\times [0,m]\rightarrow \mathbb{C}$ and separately interpret real and imaginary parts of the integral as probabilities of logical values. 
\end{rem}	

\begin{rem}
	We can consider an encoder $z(\xi;\cdot,\cdot)\in\mathcal{E}_{\mathbf{F}}(n\times m)$ as a special kind of a probabilistic cellular automaton with cells $\Delta_{tj}$, $j=\overline{1,m}$ in time $t$ and possible cell states $\{-1,0,1\}$. The function $\xi:[0,n]\rightarrow \mathbb{R}^m$ determines the probabilities of transitions between states.
\end{rem}

\begin{exa}[A Soviet Story (continued)]\label{exa:ss-c}
A PDNF can encode the possibilities from Example~\ref{exa:ss} by assigning weights to each variable in each conjunction. An example of the corresponding piecewise constant encoders $z_S \in \mathcal{E}^8_{\mathbf{F}}(3)$ and $U_S \in \mathcal{E}_{\mathbf{F}}(3\times 8)$ can be seen in Fig. \ref{fig:llp}. 

\begin{figure}[htbp]
	\centering
	\begin{subfigure}{\textwidth}
		\centering
		\resizebox{0.95\linewidth}{!}{%
			\begin{tikzpicture}[x=0.55cm, y=6cm]
				\draw[->, >=stealth, thick, gray!80] (0,0) -- (25,0);
				\draw[->, >=stealth, thick, gray!80] (0,0) -- (0,1.05);
				
				\foreach \y in {0.2, 0.4, 0.6, 0.8, 1.0} {
					\draw[gray!50] (0,\y) -- (-0.1,\y) node[left, black, font=\small\sffamily] {\y};
				}
				
				\foreach \x in {0, 8, 16, 24} {
					\draw[gray!50] (\x,0) -- (\x,-0.02) node[below, black, font=\small\sffamily, yshift=-2pt] {\x};
				}
				
				\draw[color=graphblue, very thick, line join=round] (0, 0.5) -- (1, 0.5) -- (1, 0.333) -- (2, 0.333) -- (2, 0.167) -- (3, 0.167) -- (3, 0) -- (8, 0);
				
				\draw[color=graphblue, very thick, line join=round] (8, 0) -- (11, 0) -- (11, 0.333) -- (12, 0.333) -- (12, 0.667) -- (13, 0.667) -- (13, 0) -- (16, 0);
				
				\draw[color=graphblue, very thick, line join=round] (16, 0) -- (21, 0) -- (21, 0.333) -- (24, 0.333) -- (24, 0);
				
				\tikzset{dot/.style={circle, fill=graphblue, inner sep=1.5pt}}
				\tikzset{lbl/.style={font=\scriptsize\sffamily, black}}
				
				\node[dot] at (0, 0.5) {};    \node[lbl, above] at (0.5, 0.5) {V};
				\node[dot] at (1, 0.333) {};  \node[lbl, above] at (1.5, 0.333) {I};
				\node[dot] at (2, 0.167) {};  \node[lbl, above] at (2.5, 0.167) {P};
				\foreach \x/\l in {3/W, 4/O, 5/F, 6/J, 7/D} {
					\node[dot] at (\x, 0) {}; \node[lbl, above, yshift=2pt] at (\x+0.5, 0) {\l};
				}
				
				\foreach \x/\l in {8/V, 9/I, 10/P} {
					\node[dot] at (\x, 0) {}; \node[lbl, above, yshift=2pt] at (\x+0.5, 0) {\l};
				}
				\node[dot] at (11, 0.333) {}; \node[lbl, above] at (11.5, 0.333) {W};
				\node[dot] at (12, 0.667) {}; \node[lbl, above] at (12.5, 0.667) {O};
				\foreach \x/\l in {13/F, 14/J, 15/D} {
					\node[dot] at (\x, 0) {}; \node[lbl, above, yshift=2pt] at (\x+0.5, 0) {\l};
				}
				
				\foreach \x/\l in {16/V, 17/I, 18/P, 19/W, 20/O} {
					\node[dot] at (\x, 0) {}; \node[lbl, above, yshift=2pt] at (\x+0.5, 0) {\l};
				}
				\node[dot] at (21, 0.333) {}; \node[lbl, above] at (21.5, 0.333) {F};
				\node[dot] at (22, 0.333) {}; \node[lbl, above] at (22.5, 0.333) {J};
				\node[dot] at (23, 0.333) {}; \node[lbl, above] at (23.5, 0.333) {D};
				
				\node[below, font=\small\bfseries\sffamily, text=black!80] at (12.0, -0.08) {Time-shifted variables};
				\node[left, font=\small\bfseries\sffamily, text=black!80, rotate=90, anchor=south, yshift=0.8cm] at (0, 0.5) {Probability};

			\end{tikzpicture}%
		}
		\caption{A piecewise constant encoder $z_S \in \mathcal{E}^8_{\mathbf{F}}(3)$ as a function of one variable}
		\label{fig:step_graph}
	\end{subfigure}
	\hfill
	\begin{subfigure}{\textwidth}
		\centering
		\resizebox{0.75\linewidth}{!}{%
			\begin{tikzpicture}[x=1.0cm, y=1.0cm]
				
				\fill[heatlightorange] (0,0) rectangle (1,1); 
				\fill[heatdarkorange]  (1,0) rectangle (2,1); 
				\fill[heatgrey]        (2,0) rectangle (3,1); 
				\fill[heatblue]        (3,0) rectangle (8,1); 
				
				\fill[heatblue]        (0,1) rectangle (3,2); 
				\fill[heatdarkorange]  (3,1) rectangle (4,2); 
				\fill[heatcream]       (4,1) rectangle (5,2); 
				\fill[heatblue]        (5,1) rectangle (8,2); 
				
				\fill[heatblue]        (0,2) rectangle (5,3); 
				\fill[heatdarkorange]  (5,2) rectangle (6,3); 
				\fill[heatdarkorange]  (6,2) rectangle (7,3); 
				\fill[heatdarkorange]  (7,2) rectangle (8,3); 
				
				\foreach \x in {1,2,...,7} {
					\draw[gray!40, dashed, thin] (\x, 0) -- (\x, 3);
				}
				\foreach \y in {1,2} {
					\draw[gray!40, dashed, thin] (0, \y) -- (8, \y);
				}
				
				\draw[->, >=stealth, thick, black!70] (0,0) -- (8.7,0);
				\draw[->, >=stealth, thick, black!70] (0,0) -- (0,3.6);
				\draw[black!30] (0,3) -- (8,3) -- (8,0);
				
				\foreach \x in {0,1,...,8} {
					\draw[black!70] (\x,0) -- (\x,-0.08);
				}
				\foreach \y in {1,2,3} {
					\draw[black!70] (0,\y) -- (-0.08,\y);
				}
				
				\node[below left, font=\small\sffamily, inner sep=2pt] at (0,0) {0};
				
				\foreach \x in {1,2,...,8} {
					\node[below, font=\small\sffamily, yshift=-1pt] at (\x, 0) {\x};
				}
				
				\foreach \x/\label in {0.5/V, 1.5/I, 2.5/P, 3.5/W, 4.5/O, 5.5/F, 6.5/J, 7.5/D} {
					\node[below, font=\small\bfseries\sffamily, yshift=-14pt] at (\x, 0) {\label};
				}
				
				\node[below, font=\small\bfseries\sffamily, text=black!80] at (4.0, -1.1) {Variables};
				
				\foreach \y/\t in {1, 2, 3} {
					\node[left, font=\small\sffamily, xshift=-4pt] at (0, \y) {\t};
				}
				\node[left, font=\small\bfseries\sffamily, text=black!80, rotate=90, anchor=south, yshift=0.8cm] at (0, 1.5) {Time};
				
				\begin{scope}[xshift=9.3cm, yshift=0.25cm]
					\node[anchor=south west, font=\small\bfseries\sffamily] at (0, 2.6) {Probability};
					
					\fill[heatblue]        (0, 0.0) rectangle (0.5, 0.5); \node[right, font=\small\sffamily] at (0.6, 0.25) {$0$};
					\fill[heatgrey]        (0, 0.5) rectangle (0.5, 1.0); \node[right, font=\small\sffamily] at (0.6, 0.75) {$1/6$};
					\fill[heatdarkorange]  (0, 1.0) rectangle (0.5, 1.5); \node[right, font=\small\sffamily] at (0.6, 1.25) {$1/3$};
					\fill[heatlightorange] (0, 1.5) rectangle (0.5, 2.0); \node[right, font=\small\sffamily] at (0.6, 1.75) {$1/2$};
					\fill[heatcream]       (0, 2.0) rectangle (0.5, 2.5); \node[right, font=\small\sffamily] at (0.6, 2.25) {$2/3$};
					
					\draw[gray!80] (0,0) rectangle (0.5, 2.5);
				\end{scope}
				
			\end{tikzpicture}%
		}
		\caption{A piecewise constant encoder $U_S \in \mathcal{E}_\mathbf{F}(3 \times 8)$ as a function of two variables. The colour intensity corresponds to probability weights}
		\label{fig:soviet_story_matrix}
	\end{subfigure}
	\caption{Two types of piecewise constant encoders for the ``Soviet Story'' narrative}
	\label{fig:llp}
\end{figure}
\end{exa}

\section{Structural similarity between PDNFs, classical DNFs, venjunction, and asynchronous sequential logic\label{sec:vasuk}}

\subsection{PDNFs as Generative Models}
Recall the concepts of venjunction \cite[P. 6]{vasyukevich2011digital} and sequention \cite[P. 45]{vasyukevich2011digital}, the logic-dynamical operations of asynchronous sequential logic. Venjunction is an asymmetric logical-dynamic operation, denoted by $\angle$, such that $x \angle y$ evaluates to 1 if and only if $x \wedge y = 1$ and, at the moment when $x$ becomes 1, $y$ is already 1. Sequention is defined as a binary function with arguments, which are represented in the manner of sequences of binary elements-variables like the following dependence:
\[\langle x \rangle= G(\langle x_1, x_2,\dots, x_n\rangle).\] 
All possible states of the function are $\langle x \rangle = 1$ and $\langle x \rangle = 0$. Sequention takes the value that depends on the values of variables, as well as on the order in which these values appear.

In this sense, our PDNFs can be endowed with semantics similar to that of sequences of elementary conjunctions. More precisely, a PDNF $Z(\Xi) = X_1\sqcup \dots \sqcup X_n$ with weight matrix $\Xi = (\xi_{ij}) \in \mathbb{R}^{n \times m}$, where $X_i$, $i=\overline{1,n}$ are probabilistic elementary conjunctions of $m$ variables $x_1^{\xi_{i1}},\dots, x_m^{\xi_{im}}$, can be seen as an analogue of venjunction $Z = \tilde{X}_n\angle \dots \angle \tilde{X}_1$, where $\tilde{X}_i$, $i=\overline{1,n}$ are elementary conjunctions of $0\leq m_i\leq m$ variables which are chosen from the list $x_1,\dots, x_m$. Note that the order of conjunctions in the venjunctive chain is reversed, reflecting temporal precedence, as $X_1$ is the first observation and $X_n$ is the last one. However, we have weighted literals $x_j^{\xi_{ij}}$ instead of boolean variables $x_j$, and each probabilistic elementary conjunction can be realised as $3^m$ possible elementary conjunctions. So, we can denote as $\mathcal{V}(m,n)$ the finite set of all $3^{mn}$ possible venjunctions that can be obtained by replacing each weighted literal $x_j^{\xi_{ij}}$ in the PDNF with one of the three deterministic literals $x_j$, $\overline{x}_j$, or $\varepsilon$.

\begin{defi}
Let $\Sigma=\cup_{j=1}^m\{\overline{x}_j,\varepsilon,x_j\}$ (see Subsection \ref{sec:regexp} for details). Define the literal projection function $\psi: \Sigma \times \mathbb{R} \to [0, 1]$ as:
\[
\psi(v, \gamma) = \begin{cases}
	F_{-1}^{j}(\gamma), & \text{if } v = \overline{x}_j, \\
	F_{0}^{j}(\gamma),  & \text{if } v = \varepsilon, \\
	F_{1}^{j}(\gamma),  & \text{if } v = x_j,
\end{cases}
\quad \gamma\in \mathbb{R}.
\]
\end{defi}

A PDNF $Z(\Xi)\in \mathcal{B}^m_{\mathbf{F}}(n)$ with the mapping $\mathbf{F} = (F_{-1},F_0,F_1)$ itself can be considered as a generative model \(P(\langle v \rangle \mid \Xi)\) with the observable variable $\langle v \rangle$, parametrised by the weight matrix $\Xi$. Suppose that for each position $(i,j)$ (corresponding to $j$-th variable in the $i$-th conjunction) of a PDNF $Z(\Xi)$, the functions $F_{-1}^j, F_0^j, F_1^j$ assign probabilities to the three possible realisations of $x_j^{\xi_{ij}}$ independently across assignments of different $(i,j)$ pairs, which is a natural assumption when sensor errors or random influences are independent across time and sensors. The semantics of a PDNF $Z(\Xi)$ can be formally characterised through a two-stage evaluation framework that separates structural possibility from numerical probability:
\begin{enumerate}
	\item Structural Stage (Possibilistic): a formula $Z(\Xi) \in \mathcal{B}_{\mathbf{F}}^m(n)$ defines a support mapping $\mathrm{supp}(P(\langle v \rangle \mid \Xi)) \subseteq \mathcal{V}(m, n)$ that filters out structurally inconsistent temporal trajectories. This support delineates the set of consistent worlds, acting as a Kripke frame of possible states:
	\[
	\mathrm{supp}(P(\langle v \rangle \mid \Xi)) = \big\{ \langle v \rangle \in \mathcal{V}(m, n) \;\big|\; \forall (i,j), \; \psi(v_{ij},\xi_{ij}) > 0 \big\}.
	\]
	\item Valuation Stage (Probabilistic): conditional on the weight matrix $\Xi = (\xi_{ij}) \in \mathbb{R}^{n \times m}$, the mapping $\mathbf{F}$ induces a discrete probability measure $\mu_{\Xi}$ over this support. The generation of a single actualised trajectory 
	\[\langle v\rangle = v_{n1}\wedge\dots\wedge v_{nm}\angle \dots \angle v_{11}\wedge\dots\wedge v_{1m}\]
	is governed by coordinate-wise probabilities:
	\[
	\mu_{\Xi}(\{\langle v \rangle\}) = \prod_{i=1}^{n} \prod_{j=1}^{m} \psi(v_{ij}, \xi_{ij}).
	\]
\end{enumerate}

Thus, the triple $\big(\mathcal{V}(m, n), 2^{\mathcal{V}(m, n)}, \mu_{\Xi}\big)$ constitutes a valid probability space. The easy proof of this fact can be found in Appendix \ref{app:prob}. In this semantics, the square disjunction sign $\sqcup$ can be considered as the operator indexing the temporal evolution of the agent's state of knowledge:
\[
\bigsqcup:(\mathcal{B}^m_\mathbf{F}(1))^n\rightarrow \mathbb{R}^{n\times m},\quad \bigsqcup:  (X_1,\dots,X_n)\mapsto \Xi,\quad Z(\Xi) = X_1\sqcup\dots \sqcup X_n.\]

We estimate the probability that a randomly generated venjunction $\langle v\rangle$ (according to $\mu_\Xi$) is close to a fixed target venjunction $\langle v'\rangle$ of the same length under a suitable metric and give an example of such an estimate in Appendix \ref{app:close}.

\subsection{PDNFs as Generators of Finite Languages\label{sec:regexp}}
A PDNF $Z$ with fixed weight matrix $\Xi$ can be viewed as a probabilistic generator of a finite language over the alphabet $\Sigma = \cup_{j=1}^m\{\overline{x}_j, \varepsilon, x_j\}$ (treating $\varepsilon$ as a blank symbol in the alphabet). Each realisation $v$ is a word of length $nm$ (consisting of $n$ consecutive blocks of length $m$, each corresponding to a time step) over this alphabet, and the set of all possible words is
\[
L(Z) = \{ \langle v\rangle \in \mathcal{V}(m,n) \mid \mu_\Xi(\{\langle v\rangle\}) > 0 \}.
\]
Thus $L(Z)$ is the support of the probability measure $\mu_\Xi$, and its size $|L(Z)|$ can range from $1$ (deterministic case) up to $3^{nm}$.

As the choices at different positions are independent, $L(Z)$ is a Cartesian product of local languages:
\[
L(Z) = \prod_{i=1}^n \prod_{j=1}^m L_{ij},
\]
where each $L_{ij}$ is the set of literals that can appear at position $(i,j)$, i.e.
\[
L_{ij} = \{ \ell \in \{\overline{x}_j,\varepsilon,x_j\} \mid F_\ell^j(\xi_{ij}) > 0 \}.
\]
Consequently, L(Z) is a finite factorised (rectangular) regular language, and its regular expression is simply the concatenation over $i,j$ of the alternatives allowed at each position.

If we interpret the PDNF as a template for generating plausible hypotheses, then $L(Z)$ is precisely the set of hypotheses that are considered. A target statement $v'$ belongs to $L(Z)$ iff for every $(i,j)$ the literal $v'_{ij}$ is allowed by the corresponding weight, i.e. $\psi(v'_{ij}, \xi_{ij}) > 0$.  
Thus, the choice of weights defines a filter on the space of all possible venjunctions.

Now suppose we do not know the weights but observe independent draws from $\mu_\Xi$.  
As discussed in Section~\ref{sec:automata}, after sufficiently many observations we can, with high confidence, identify the whole set $L(Z)$ (see Proposition \ref{prop:regexp}).  

Once $L(Z)$ is known, we have an explicit enumeration of all possible behaviours of the system.  
If we also wish to know which of these behaviours are actually useful or true, further testing or logical reasoning is required, as the PDNF itself only provides the candidate set.

The language $L(Z)$ can be described by a probabilistic regular expression of the form
\[
\bigwedge_{i=1}^n \bigwedge_{j=1}^m ( \overline{x}_j \mid \varepsilon \mid x_j )\quad \text{but with some alternatives removed},
\]
where the removal is determined by the condition $F_\ell^j(\xi_{ij}) > 0$.

\section{Probability Measure on the Space of Encoders\label{sec:newsec}}

\subsection{On Continuous Logical Expressions}

\begin{figure}[htbp]
	\centering
	\resizebox{\linewidth}{!}{%
		\begin{tikzpicture}[x=0.09cm, y=4.5cm, line join=round, line cap=round]
			
			\draw[gray!10, dashed, ystep=0.2, xstep=10] (0,0) -- (155,1.1);
			
			\draw[->, >=stealth, thick, black!80] (0,0) -- (165,0) node[below, black, font=\small\sffamily, xshift=-45pt, yshift=-12pt] {$\tilde{j} + \mathrm{Time}$};
			\draw[->, >=stealth, thick, black!80] (0,0) -- (0,1.1) node[left, black, font=\small\sffamily, rotate=90, anchor=south, yshift=16pt, xshift=-30pt] {$\mathrm{Confidence}$};
			
			\foreach \y in {0.2, 0.4, 0.6, 0.8, 1.0} {
				\draw[black!60] (0,\y) -- (-2pt,\y) node[left, black, font=\small\sffamily] {\y};
			}
			
			\foreach \x in {0, 30, 50, 90, 100, 135, 150} {
				\draw[black!60] (\x,0) -- (\x,-2pt) node[below, black, font=\small\sffamily, yshift=-2pt] {\x};
			}
			
			\draw[color=blue!70!black, very thick, smooth, samples=30, domain=20:40] 
			plot (\x, {0.5*(1 + cos((\x - 30)*18))});
			\draw[color=blue!70!black, very thick] (0,0) -- (20,0);
			
			\draw[color=blue!70!black, very thick, smooth, samples=30, domain=80:100] 
			plot (\x, {0.5*(1 + cos((\x - 90)*18))});
			\draw[color=blue!70!black, very thick] (40,0) -- (80,0);
			
			\draw[color=blue!70!black, very thick, smooth, samples=30, domain=125:145] 
			plot (\x, {0.5*(1 + cos((\x - 135)*18))});
			\draw[color=blue!70!black, very thick] (100,0) -- (125,0);
			\draw[color=blue!70!black, very thick] (145,0) -- (155,0);
			
			\draw[color=red!70!black, thick, dashed, smooth, samples=30, domain=22:42] 
			plot (\x, {0.5*(1 + cos((\x - 32)*18))});
			\draw[color=red!70!black, thick, dashed] (0,0) -- (22,0);
			
			\draw[color=red!70!black, thick, dashed, smooth, samples=30, domain=78:98] 
			plot (\x, {0.5*(1 + cos((\x - 88)*18))});
			\draw[color=red!70!black, thick, dashed] (42,0) -- (78,0);
			
			\draw[color=red!70!black, thick, dashed, smooth, samples=30, domain=127:147] 
			plot (\x, {0.5*(1 + cos((\x - 137)*18))});
			\draw[color=red!70!black, thick, dashed] (98,0) -- (127,0);
			\draw[color=red!70!black, thick, dashed] (147,0) -- (155,0);
			
			\tikzset{lbl/.style={font=\small\sffamily, fill=white, inner sep=1pt}}
			
			\node[lbl, text=black, above] at (31,1.02) {Year 1};
			\node[lbl, text=black, above] at (89,1.02) {Year 2};
			\node[lbl, text=black, above] at (136,1.02) {Year 3};
			
		\end{tikzpicture}
	}
	\caption{Comparison of two continuous functional encoders $z_1, z_2 \in \mathcal{E}^{[0,50]}_{\mathbf{F}}(3)$ representing Peter's weight measurements from two different scales over three consecutive years. The solid and dashed lines reflect the independent epistemic confidence profiles of the scales, aligned on the unified axis ($\tilde{j}+ \mathrm{Time}$).}
	\label{fig:peter_weight_bump}
\end{figure}

	A critical question may arise regarding the necessity of invoking the machinery of measure theory and functional analysis if the encoders of a PDNF, such as those in Example~\ref{exa:ss}, Example~\ref{exa:ss-c}, appear to be merely piecewise constant functions. When logical propositions concern strictly discrete atomic entities (e.g., identity checks like ``Peter'' or ``Ivan''), a finite-dimensional matrix representation may indeed suffice. However, this discrete approach breaks down when propositions involve continuous attributes (e.g., temperature or mass). For instance, in an assertion like ``Peter weighs 80~kg'', a microscopic measurement perturbation (e.g., ``Peter weighs 80~kg and 1~g'') should intuitively be treated as negligible informational noise rather than a logical contradiction, whereas an error of 10~kg (i.e., ``Peter weighs 90~kg'') yields a false statement. To model this accurately, the subjective confidence of an agent should not be encoded via discontinuous step functions, but rather via smooth, compactly supported ``bump'' functions (smooth, non-negative functions that are strictly zero outside a bounded interval of measurement tolerance) centred at the most accurate empirical point.
	
	Assume that Peter measures his weight, which remains within the interval of $[50, 100]$~kg, consecutively for three years. Figure~\ref{fig:peter_weight_bump} illustrates the encoding of the following sequential continuous assertion: ``In the first year, Peter is convinced that his weight is approximately $80$~kg; in the second year, he is convinced it is approximately $90$~kg; and in the third year, he is convinced it is approximately $85$~kg''. 

	An interval of length $50$ units on the abscissa axis is allocated to the continuous variable ``Peter's weight in the year $t$ is $j$'' $m(t,j)$, which is periodically replicated in accordance with the chronological order of the measurement steps $t=1,2,3$. Because Peter's weight $j$ is bounded between $50$ and $100$~kg, to eliminate redundant empty domains, the encoder parametrises the variable $\tilde{j} = j - 50$ instead of the raw variable $j$.

	The subjective confidence takes values in $[0,1]$, with $1$ corresponding to maximal confidence. As a function mapping parameter values to $[0,1]$, it acts as an unnormalised likelihood profile whose spread reflects measurement uncertainty. Consequently, the subjective confidence $\mathrm{Confidence}_1$ of the agent manifests as a sequence of smooth bump functions shifted smoothly along the unified $\tilde{j} + \mathrm{Time}$ axis. The PDNF for this example contains an infinite number of variables and could be written as
	\[
	Z_1=\bigsqcup_{t=1}^3\bigwedge_{j\in [50,100]}m(t,j)^{\mathrm{Confidence}_1(t,j)}.
	\]	
	Denote the corresponding encoder as $z_1$ (the solid line in Figure~\ref{fig:peter_weight_bump}).
	
	Suppose that Peter also uses other scales to check the measurements, and these scales give slightly different weights and so different beliefs for Peter (the dashed line):
	\[
	Z_2=\bigsqcup_{t=1}^3\bigwedge_{j\in [50,100]}m(t,j)^{\mathrm{Confidence}_2(t,j)}.
	\]	
	Denote the corresponding encoder as $z_2$ (the dashed line in Figure~\ref{fig:peter_weight_bump}). Also, assume that the functions $\mathrm{Confidence}_k(t,\cdot)$ are not identically zero on $[50,100]$ for $k=1,2$ and $t=1,2,3$.
	
	We denote the space of these encoders as $\mathcal{E}^{[0,50]}_{\mathbf{F}}(3)$ (we could use $\mathcal{E}^{\infty}_{\mathbf{F}}(3)$ which is closer to our previous definition but $\mathcal{E}^{[0,50]}_{\mathbf{F}}(3)$ seems more precise). Such continuous propositions cannot be operated upon or compared with one another using any classical finite-dimensional methods; however, their corresponding functional encoders can be easily compared by means of norms within the encoder space $\mathcal{E}^{[0,50]}_{\mathbf{F}}(3)$ as $\delta = \|z_1-z_2\|_1$. If this $\delta$ is too large, Peter can conclude that one of the scales is broken.

To define these continuous logical expressions and avoid the semantic collapse caused by infinite pointwise intersections, we treat them as formal indexed syntactic constructs whose semantics is defined directly via induced probability measures over measurable events, rather than via pointwise truth assignments.

\begin{defi}[Spatial continuous PDNF]\label{def:SPDNF}
	Let $J \subset \mathbb{R}$ be a closed interval equipped with its Borel $\sigma$-algebra $\mathcal{B}(J)$, and let $\mathcal{L}$ be a Banach space of weight functions $\xi_t: J \to \mathbb{R}$ (e.g., $L_1(J)$). A \emph{spatial continuous PDNF} of $n$ time steps is a formal expression of the form
	\[
	Z = \bigsqcup_{t=1}^n \bigwedge_{j \in J} x_t(j)^{\xi_t(j)},
	\]
	equipped with a family of mappings $\mathbf{F}_t: \mathcal{L} \to \mathcal{M}_+(J)^3$, $t=\overline{1,n}$, where $\mathcal{M}_+(J)$ denotes the cone of finite non-negative Borel measures on $(J, \mathcal{B}(J))$. For each $t$ and $\xi_t \in \mathcal{L}$, $\mathbf{F}_t(\xi_t) = (\mu_{t,-1}^{\xi_t}, \mu_{t,0}^{\xi_t}, \mu_{t,1}^{\xi_t})$ satisfies the global normalisation condition
	\[
	\mu_{t,-1}^{\xi_t}(J) + \mu_{t,0}^{\xi_t}(J) + \mu_{t,1}^{\xi_t}(J) = 1.
	\]
	Semantically, $\mathbf{F}_t(\xi_t)$ defines a probability measure on the step-$t$ outcome space $J \times \{-1, 0, 1\}$. For any measurable event $B \in \mathcal{B}(J)$, the induced functional $F_{t,1}(\xi_t; B) := \mu_{t,1}^{\xi_t}(B)$ gives the probability that the true parameter at step $t$ lies in $B$ with positive truth status (while $\mu_{t,-1}^{\xi_t}(B)$ and $\mu_{t,0}^{\xi_t}(B)$ give the probabilities of negation and indeterminacy on $B$, respectively). Assuming independence across steps $t=\overline{1,n}$, the joint distribution on $(J \times \{-1,0,1\})^n$ is the product measure $\bigotimes_{t=1}^n \mathbf{F}_t(\xi_t)$.
\end{defi}

In the weight-measurement example above ($J = [50, 100]$, $n = 3$), we set $\mu_{t,-1}^{\xi_t} \equiv 0$, $\mu_{t,0}^{\xi_t} \equiv 0$, and define $\mu_{t,1}^{\xi_t}$ via its Lebesgue density
\[
\mathrm{PDF}_{m}(t,j) = \left(\int_{50}^{100} \mathrm{Confidence}_k(t,u) \, du\right)^{-1} \mathrm{Confidence}_k(t,j),
\]
so that for any test interval $(j_0, j_1) \subset [50, 100]$, the probability that Peter's weight at year $t$ lies in $(j_0, j_1)$ is given by the functional
\[
F_{t,1}\bigl(\mathrm{Confidence}_k(t,\cdot); (j_0, j_1)\bigr) = \mu_{t,1}^{\xi_t}\bigl((j_0, j_1)\bigr) = \int_{j_0}^{j_1} \mathrm{PDF}_{m}(t,j) \, dj.
\]

Next, consider continuous time evolution. Let us model the evolution over time of Peter's confidence that his weight is $85$~kg via the formal expression
\[
Z_{E} = \bigsqcup_{t \in [0,24)} m(t,85)^{\mathrm{Confidence}(t,85)},
\]
where $t \in [0,24)$ denotes a continuous time horizon (e.g., a 24-month period). Here, the corresponding functional encoder resides in $\mathcal{E}^{1}_{\mathbf{F}}([0,24))$ and can be compared with other encoders via Banach space norms. To endow continuous-time expressions with a well-defined probability measure on trajectories, we introduce the following definition.

\begin{defi}[Temporal continuous PDNF]\label{def:TPDNF}
	Let $I \subset \mathbb{R}$ be a closed time interval and let $\mathcal{T}$ be a Banach space of vector-valued weight trajectories $\xi = (\xi_1, \dots, \xi_m): I \to \mathbb{R}^m$ (e.g., $L_2(I)^m$). Let $(\Omega_I, \mathcal{F}_I)$ be a measurable space of admissible literal trajectories $\omega: I \to \{-1, 0, 1\}^m$, and let $\mathcal{P}(\Omega_I)$ denote the set of probability measures on $(\Omega_I, \mathcal{F}_I)$. A \emph{temporal continuous PDNF} is a formal expression of the form
	\[
	Z = \bigsqcup_{t \in I} \bigwedge_{j=1}^m x_j(t)^{\xi_j(t)},
	\]
	equipped with a weight mapping $\mathbf{F}: \mathcal{T} \to \mathcal{P}(\Omega_I)$ that assigns to each weight trajectory $\xi \in \mathcal{T}$ a probability measure $\mu_\xi = \mathbf{F}(\xi)$ on $(\Omega_I, \mathcal{F}_I)$.
\end{defi}

In the continuous-time example $Z_E$ ($m=1$, $I=[0,24)$), let $\xi(t) = \mathrm{Confidence}(t,85)$, and let $\Omega_I \subset L_2([0,24))$ be the measurable space of binary truth-value trajectories $\omega: [0,24) \to \{0,1\}$ (representing when the statement ``Peter's weight at time $t$ is $85$'' holds), equipped with a reference probability measure $\mu_0$ (e.g., induced by a continuous-time Markov jump process). The mapping $\mathbf{F}$ can then be specified via a Gibbs-type Radon--Nikodym derivative with respect to $\mu_0$:
\[
\frac{d\mu_\xi}{d\mu_0}(\omega) = \frac{1}{\mathcal{Z}(\xi)} \exp\left( -\int_{0}^{24} \bigl(\omega(t) - \xi(t)\bigr)^2 \, dt \right),
\]
where $\mathcal{Z}(\xi) = \int_{\Omega_I} \exp\bigl(-\|\omega - \xi\|_{L_2}^2\bigr) \, d\mu_0(\omega)$ is the partition function. For any measurable set of trajectories $A \in \mathcal{F}_I$, the probability of observing a trajectory in $A$ is $\mu_\xi(A) = \int_A \frac{d\mu_\xi}{d\mu_0}(\omega) \, d\mu_0(\omega)$.

More generally, let an agent's statement be formalised as an attributed belief graph $\Gamma$, where vertices represent the subjects of discourse and edges denote their relational dependencies. If these vertices and edges are labelled with elements of a continuous metric space (such as sensor data or continuous semantic attributes), one can rigorously establish a metric of dissimilarity, for example, via a graph edit distance that incorporates the metrics of the underlying labels. This geometric structure induces a continuous epistemic neighbourhood (a ``cloud'' of plausible worlds) surrounding a given belief graph \cite{KuznetsovLemtyuzhnikova2026, kuznetsov2026modelling}. From this epistemic perspective, continuous PDNFs retain the status of logical forms rather than mere probability measures: they represent structured logical propositions whose truth-value assignments are continuously ``smeared'' across an agent's belief neighbourhood around a classical logical assertion. By defining the PDNF encoders as integrable functions within a Banach space, we ensure that continuous variations in the attributes map continuously (for continuous weight mappings $\mathbf{F}$) to variations of the induced probability measures. Combining both continuous dimensions yields the general spatio-temporal case.

\begin{defi}[Spatio-temporal continuous PDNF]\label{def:STPDNF}
	Let $I, J \subset \mathbb{R}$ be closed intervals, and let $\mathcal{G}$ be a Banach space of spatio-temporal weight fields $\xi: I \times J \to \mathbb{R}$ (e.g., $L_2(I \times J)$). Let $(\Omega_{I \times J}, \mathcal{F}_{I \times J})$ be a measurable space of spatio-temporal realisations over $I \times J \times \{-1,0,1\}$. A \emph{spatio-temporal continuous PDNF} is a formal expression of the form
	\[
	Z(\xi) = \bigsqcup_{t \in I} \bigwedge_{j \in J} x(t,j)^{\xi(t,j)},
	\]
	equipped with a mapping $\mathbf{F}: \mathcal{G} \to \mathcal{P}(\Omega_{I \times J})$ that assigns to each weight field $\xi \in \mathcal{G}$ a probability measure $\mu_\xi = \mathbf{F}(\xi)$ on $(\Omega_{I \times J}, \mathcal{F}_{I \times J})$.
\end{defi}

\subsection{Probability Measure on the Space of Encoders\label{ssec:probmeasure}}
In the sensor network applications described in Section~\ref{sec:Markov} and Appendix~\ref{app:samples}, the weight matrix $\Xi$ may itself be random, depending on unobserved physical quantities. In that case, one obtains a two‑stage random process: at first, a random matrix $\Xi$ is generated, and then, conditionally on $\Xi$, a venjunction is drawn according to $\mu_{\Xi}$. In terms of encoders from $\mathcal{E}^m_{\mathbf{F}}(n)$, we draw a random element (i.e., a function) from the set $\mathcal{E}^m_{\mathbf{F}}(n)$ and then encode a venjunction with the function drawn.

The overall distribution on $\mathcal{V}(m,n)$ is then given by the Lebesgue integral 
\[
\mathbb{P}(A) = \int_{\mathbb{R}^{n\times m}} \mu_{\Xi}(A) \, dP(\Xi),\qquad A \subseteq \mathcal{V}(m,n),
\]
where $P$ denotes the distribution of $\Xi$. We should note that the distribution of $\Xi$ can be continuous even in the case of finite number of variables in a PDNF. This layered construction allows the PDNF formalism to express both the uncertainty in the underlying parameters (encoded in $\Xi$) and the stochastic realisation of the logical observations (encoded by $\mathbf{F}$).

The construction of Subsection~\ref{ss:cont} associates to each weight matrix $\Xi\in\mathbb{R}^{n\times m}$ a function $z_\Xi\in\mathcal{E}^m_{\mathbf{F}}(n)$, for instance the piecewise constant encoder
\[
z_\Xi(x)=z(\xi;x)=\sum_{i=1}^n\sum_{j=1}^m \xi_{ij}\mathbf{1}_{\Delta_{ij}}(x),
\]
where $\xi:\{1,\dots, n\}\rightarrow \mathbb{R}^m$, $\xi_j(i) = \xi_{ij}$. If $\Xi$ is random with distribution $P$ on $\mathbb{R}^{n\times m}$, then the image measure $Q = \Phi_*P$ defines a probability distribution on $\mathcal{E}^m_{\mathbf{F}}(n)$. Here $\Phi_*P$ is the pushforward of $P$ under the map $\Phi:\Xi\mapsto z_\Xi$, $(\Phi_*P)(A)= P(\Phi^{-1}(A))$.

For a fixed encoder $z\in \mathcal{E}^m_{\mathbf{F}}(n)$, the probabilities of venjunctions are obtained from the integrals as in (\ref{eq:cont01})--(\ref{eq:cont03}). Let $\mu_z$ denote the corresponding probability measure on $\mathcal{V}(m,n)$. The mapping $z\mapsto\mu_z$ takes values in the finite-dimensional Banach space $\mathcal{M}(\mathcal{V}(m,n))$ of finite signed measures on $\mathcal{V}(m,n)$ equipped with the total variation norm
\[
\|\mu\|_{TV} = \sum_{\langle v\rangle \in \mathcal{V}(m,n)} |\mu(\{\langle v\rangle\})|.
\]
Since $\mathcal{V}(m,n)$ is finite, $\mathcal{M}(\mathcal{V}(m,n))$ is separable, and if $z\mapsto\mu_z(A)$ is measurable for every $A\subseteq\mathcal{V}(m,n)$ and $\int_{\mathcal{E}^m_{\mathbf{F}}(n)} \|\mu_z\|_{TV}\,dQ(z) = 1 < \infty$, the measure-valued Bochner integral \cite{diestel1977vector}
\begin{equation}\label{int:Bochner}
	\mathbb{P} = \int_{\mathcal{E}^m_{\mathbf{F}}(n)} \mu_z\,dQ(z) \in \mathcal{M}(\mathcal{V}(m,n))
\end{equation}
exists. Applying the continuous linear evaluation functional $\mu \mapsto \mu(A)$ yields the scalar Lebesgue integral
\[
\mathbb{P}(A) = \int_{\mathcal{E}^m_{\mathbf{F}}(n)} \mu_z(A)\,dQ(z), \qquad A \subseteq \mathcal{V}(m,n),
\]
which defines a probability measure on $\mathcal{V}(m,n)$. Recall that the Bochner integral generalises the Lebesgue integral to functions taking values in Banach spaces.

Thus, the two‑stage random process can be recast as first drawing a random encoder $z$ according to $Q$, and then drawing a venjunction according to $\mu_z$. This viewpoint unifies the discrete and continuous aspects of the PDNF formalism and opens the way to applying functional‑analytic tools such as Bochner integration, stochastic processes in function spaces, and ergodic theory. We will continue this discussion in Section \ref{sec:Markov}. 

\begin{rem}[Two Sources of Continuity]
	Depending on the source of continuity, the expectation of a random PDNF encoder yields two distinct regimes of Bochner integration:
	\begin{enumerate}
		\item Finite-dimensional parameter randomness (Measure-valued Bochner integral). When the PDNF has finitely many variables ($m, n < \infty$) and the weight matrix $\Xi$ follows an absolutely continuous distribution $P$ on $\mathbb{R}^{n \times m}$ (i.e., pushforward measure $Q$ on $\mathcal{E}^m_{\mathbf{F}}(n)$), the mapping $z \mapsto \mu_z$ takes values in the finite-dimensional Banach space of signed measures $\mathcal{M}(\mathcal{V}(m,n))$ equipped with the total variation norm. Here, the Bochner integral (\ref{int:Bochner}) averages discrete probability distributions over finite venjunctions, serving as a probabilistic analogue of the centroid for type-2 fuzzy sets.
		\item Infinitary ``spatio-temporal'' randomness (Function-valued Bochner integral). When the PDNF itself is defined over continuous index domains $I \times J$ (such as the spatial, temporal, and spatio-temporal expressions governed by weight functions $\xi_t \in \mathcal{L}$, $\xi \in \mathcal{T}$, or weight fields $\xi \in \mathcal{G}$ in Definitions~\ref{def:SPDNF}--\ref{def:STPDNF}), a random realisation of these weight fields induces a probability measure $Q$ directly on an infinite-dimensional Banach space of encoder functions (e.g., $\mathcal{E}_{\mathbf{F}}(I \times J) \subseteq L_2(I \times J)$). In this setting, the Bochner integral
		\[
		\mathbb{E}[z] = \int_{\mathcal{E}_{\mathbf{F}}(I \times J)} z \, dQ(z)
		\]
		is computed directly in the infinite-dimensional encoder space $\mathcal{E}_{\mathbf{F}}(I \times J)$, yielding a mean continuous confidence surface.
	\end{enumerate}
\end{rem}

\begin{rem}
	To quantify the uncertainty encoded in a PDNF, we can introduce the entropy‑based measure.
	
	For a fixed encoder $z_\Xi$ (and hence fixed probabilities $p_{tj}^{(\ell)} = F_\ell^j(\xi_{tj})$), the Shannon entropy at each position $(t,j)$ is $H_{tj} = -\sum_{\ell\in\{-1,0,1\}} p_{tj}^{(\ell)} \log p_{tj}^{(\ell)}$. Summing over all positions gives the total entropy
	\[
	H_{\text{tot}}(\Xi) = \sum_{t=1}^{n}\sum_{j=1}^{m} H_{tj},
	\]
	which vanishes for a classical DNF (all probabilities are $0$ or $1$) and is maximal when the three outcomes are equally probable.
\end{rem}

\section{Probabilistic Disjunctive Normal Forms as Automata Arrays\label{sec:automata}}
In the preceding sections, we implicitly assumed that the weight matrix~$\Xi$ was known a priori, reducing the problem to the generation of temporal venjunctions or continuous functional signals according to a defined generative model. However, in practical applications, the true inputs~$\mathcal{X}$ and the internal states of the devices remain hidden from the observer, leading to an inverse problem of structural identification. By representing an array of~$m$ stochastic automata observed over~$n$ discrete time steps as a unified PDNF, we compress a space of~$3^{nm}$ potential temporal trajectories into a weight matrix~$\Xi \in \mathbb{R}^{n \times m}$. This representation allows us to bypass combinatorial path-tracking, replacing it with the operations of a linear vector space where the accumulation of evidence from multiple observational series reduces to matrix addition.

Let us recall the definition of a deterministic finite automaton. It is the sextuple 
\[(S, \mathcal{X}, \mathcal{Y}, s^0, \delta , \lambda),\]
where $S$ is the set of internal states, $\mathcal{X}$ is the input alphabet, $\mathcal{Y}$ is the output alphabet, $s^0\in S$ is the initial state, $\delta: S\times \mathcal{X} \rightarrow S$ is the state-transition function and $\lambda$ is the output function.  If the output function acts as $\lambda: S\times \mathcal{X}\rightarrow \mathcal{Y}$, the automaton is called a Mealy machine; if the output function acts as $\lambda: S\rightarrow \mathcal{Y}$, the automaton is called a Moore machine.

Imagine that we have automata $A_j = (S_j, \mathcal{X}_j,\mathcal{Y}_j, s_j^0, \delta_j, \lambda_j)$, $j=\overline{1,m}$, which produce outputs at time moments $t_i$, $i=\overline{1,n}$. The output alphabet $\mathcal{Y}$ consists of three symbols ``$x$'', ``$\overline{x}$'', and $\varepsilon$ (the empty string). The input alphabet $\mathcal{X}$ is a subset of $\mathbb{R}^p$. Suppose that $\mathcal{X}_j = \cup_{i=1}^{q_j} \mathcal{X}_j^i$ and $\mathcal{X}_j^i\cap \mathcal{X}_j^k = \emptyset$, $i,k=\overline{1,q_j}$, $S_j=\{s_j^0,s_j^1,\dots, s_j^{q_j}\}$. If the automaton $A_j$ receives $x\in \mathcal{X}_j^i$, then $A_j$ goes to the state $s_j^i$ regardless of the previous state. The set of states is also partitioned into the three non-intersecting subsets: $S_j = S_j^{+1}\cup S_j^{-1}\cup S_j^0$, $\lambda_j(s) = x$, $s\in S_j^{+1}$, $\lambda_j(s) = \overline{x}$, $s\in S_j^{-1}$, $\lambda_j(s) = \varepsilon$, $s\in S_j^{0}$. Let $\mathcal{X}_j^{+1}$ denote the union of all $\mathcal{X}_j^{i}$ such that $s_j^i \in S^{+1}_j$, $\mathcal{X}_j^{-1}$ the union of all $\mathcal{X}_j^{i}$ that $s_j^i \in S^{-1}_j$, $\mathcal{X}_j^{0}$ the union of all $\mathcal{X}_j^{i}$ that $s_j^i \in S^{0}_j$. 

Let $(\mathcal{X}_j, \mathcal{F}_{\mathcal{X}_j}, \nu_j)$ be a probability space describing the distribution of unknown inputs to automaton $A_j$, where each partition element $\mathcal{X}_j^i \in \mathcal{F}_{\mathcal{X}_j}$ is measurable. Consequently, the subsets $\mathcal{X}_j^{+1}$, $\mathcal{X}_j^{-1}$, and $\mathcal{X}_j^{0}$ are measurable, and the probabilities of the outputs ``$x$'', ``$\overline{x}$'', and $\varepsilon$ for automaton $A_j$ are given by
\[
P_j^{+1} = \nu_j(\mathcal{X}_j^{+1}), \qquad P_j^{-1} = \nu_j(\mathcal{X}_j^{-1}), \qquad P_j^{0} = \nu_j(\mathcal{X}_j^{0}).
\]
In particular, if $\mathcal{X}_j$ is a finite set (or a bounded measurable subset of $\mathbb{R}^p$ with positive finite Lebesgue measure) and inputs are uniformly distributed, these probabilities reduce to the ratios $P_j^{\ell} = |\mathcal{X}_j^{\ell}| / |\mathcal{X}_j|$ for $\ell \in \{-1, 0, 1\}$, where $|\cdot|$ denotes cardinality (respectively, Lebesgue measure). If, moreover, $\mathcal{X}_j$ is finite and each input corresponds to exactly one state, then $\tilde{P}_j^{\ell} = |S_j^{\ell}| / |S_j|$ for $\ell \in \{-1, 0, 1\}$.

If inputs are unknown, we can estimate the probability of an output according to the aforementioned formulae and symbolically write the automaton's output as $x^\xi$, where $\xi$ is a number satisfying (\ref{eqPF1})--(\ref{eqPF3}) for some function $\mathbf{F}$ and corresponding probabilities $P_j^{+1}$, $P_j^{0}$, $P_j^{-1}$. If we have a system of $m$ automata (e.g., sensors) with unknown inputs at some time moment $t\geq 1$, we can symbolically denote it as the probabilistic elementary conjunction $x_1^{\xi_1}\wedge \dots \wedge x_m^{\xi_m}$, where $\xi_1,\dots,\xi_m$ are calculated as before. 

Let now the new state $s_j(t)$ of the automaton $A_j$ depend on the input $x(t)\in \mathcal{X}_j^i$ and on the previous state $s_j(t-1)$. In this case, we should change the partitions $\mathcal{X}_j = \cup_{i=1}^{q_j} \mathcal{X}_j^i(t)$ and $S_j(t) = S_j^{+1}(t)\cup S_j^{-1}(t)\cup S_j^0(t)$ at any time moment, and probabilities $P_j^{+1}$, $P_j^{0}$, $P_j^{-1}$ will be the functions of time $t$ and, therefore, the numbers $\xi_1(t),\dots, \xi_m(t)$ will be also the functions of $t$. We can symbolically represent the automaton, which is functioning at time moments $t=\overline{1,n}$, as the probabilistic DNF from $\mathcal{B}^m_\mathbf{F}(n)$
\[x_1^{\xi_1(1)}\wedge \dots \wedge x_m^{\xi_m(1)}\sqcup\dots  \sqcup x_1^{\xi_1(n)}\wedge \dots \wedge x_m^{\xi_m(n)}.\]

If sensors are observed by different agents, with each agent observing them over several operating cycles, then the probabilistic DNFs corresponding to each agent's observations can be aggregated according to the formula (\ref{eqsum}). This aggregation provides a combined estimate of the automata being observed and enables the formulation of additional hypotheses regarding their unknown inputs. This problem is also related to the well-known problem of automata identification (see, for example, \cite{Wharton1974}).

\begin{rem}\label{rem:ident_zero}
	Unlike the generative setting of Remark~\ref{rem:zero_weight}, the inverse problem of automata identification does not assume prior knowledge of the mapping $\mathbf{F}$ governing the device. Instead, the exponential family in Theorem~\ref{thm:composition} and Appendix~\ref{app:fusions} serves as a posterior exponential parametrisation for aggregating empirical observations over a fixed set of $m$ variables and $n$ time steps. In this regime, $F_\ell^j(0) = 1/3$ for all $\ell \in \{-1,0,1\}$, so a zero weight $\xi_{ij} = 0$ represents an uninformative uniform prior (absence of collected evidence) rather than structural absence of the variable, ensuring that adding a zero-weighted PDNF acts as a neutral identity element ($Z + 0 = Z$) under Bayesian evidence fusion.
\end{rem}

\begin{defi}[Bayesian fusion for ternary outcomes]\label{def:fusion}
	Let $p = (p_{-1}, p_0, p_1)$ and $q = (q_{-1}, q_0, q_1)$ be two probability 
	vectors on the ternary set $\{\overline{x}, \varepsilon, x\}$, representing 
	independent pieces of evidence. Their \emph{Bayesian fusion} (assuming a 
	uniform prior) is the probability vector $p \ast q$ defined by
	\[
	(p \ast q)_\ell = \frac{p_\ell \cdot q_\ell}
	{\sum_{\ell' \in \{-1,0,1\}} p_{\ell'} \cdot q_{\ell'}}, 
	\quad \ell \in \{-1,0,1\}.
	\]
\end{defi}

\begin{thm}[Composition of independent observations]\label{thm:composition}
	Let $A$ be a deterministic finite automaton with $m$ output variables 
	$x_1,\dots,x_m$, observed by two independent agents $ag_1$ and $ag_2$ over 
	$n$ synchronised time steps. Suppose each agent records its observations 
	as a PDNF $Z_k \in \mathcal{B}^m_{\mathbf{F}}(n)$ ($k=1,2$), using the same family $\mathbf{F}$ of probability 
	maps. Also, suppose that 
	\[
	F_{\ell}^j(\xi) = \frac{\exp(\alpha_{\ell}^j \xi)}
	{\sum_{\ell'\in \{-1,0,1\}} \exp(\alpha_{\ell'}^j \xi)}
	\]
	and parameters $\alpha_\ell^j \in \mathbb{R}$ are identical for both agents. Then the PDNF $Z = Z_1 + Z_2$ (with addition defined by (\ref{eqsum})) 
	represents the combined evidence from both agents, in the sense that for 
	every $i=\overline{1,n}$ and $j=\overline{1,m}$ the weight $\xi_{ij}$ in $Z$ satisfies
	\[
	F_{\ell}^j(\xi_{ij}) \propto 
	F_{\ell}^j(\xi_{ij}^{(1)}) \cdot F_{\ell}^j(\xi_{ij}^{(2)}), 
	\quad \ell \in \{-1,0,1\},
	\]
	where $\xi_{ij}^{(k)}$ are the weights in $Z_k$. 
	
	Moreover, if the agents' observations are conditionally independent given 
	the true state of $A$, then the probability vector
	\[
	\bigl(F_{-1}^j(\xi_{ij}), F_0^j(\xi_{ij}), F_1^j(\xi_{ij})\bigr)
	\]
	is exactly the Bayesian fusion of the corresponding vectors from $Z_1$ and $Z_2$.
\end{thm}

\begin{proof}
	Let us fix a time step $i$ and a variable $x_j$. For agent $k$ ($k=1,2$), 
	the weight $\xi_{ij}^{(k)}$ encodes, through the map $\mathbf{F}$, the 
	probabilities of the three possible realisations of $x_j$ at time $i$ based 
	on the agent's partial information. By Definition~\ref{def:def01}, these probabilities are
	\[
	p_{\ell}^{(k)} = F_{\ell}^j(\xi_{ij}^{(k)}), \quad \ell \in \{-1,0,1\},
	\]
	with $p_{-1}^{(k)}+p_0^{(k)}+p_1^{(k)}=1$.
	
	Under the assumption that the agents' observations are independent (given the true state of the automaton), the joint probability that both agents would assign to a particular realisation $\ell$ is proportional to the product $p_{\ell}^{(1)} \cdot p_{\ell}^{(2)}$. Normalising gives
	\[
	p_{\ell}^{\text{joint}} = 
	\frac{p_{\ell}^{(1)} p_{\ell}^{(2)}}
	{\sum_{\ell' \in \{-1,0,1\}} p_{\ell'}^{(1)} p_{\ell'}^{(2)}}.
	\]
	
	Now consider the weight $\xi_{ij} = \xi_{ij}^{(1)} + \xi_{ij}^{(2)}$ appearing 
	in the sum $Z = Z_1 + Z_2$. By the definition of the addition in 
	$\mathcal{B}^m_{\mathbf{F}}(n)$, we have
	\begin{multline*}
	F_{\ell}^j(\xi_{ij}) = F_{\ell}^j(\xi_{ij}^{(1)} + \xi_{ij}^{(2)})= \frac{\exp(\alpha_{\ell}^j (\xi_{ij}^{(1)} + \xi_{ij}^{(2)}))}
	{\sum_{\ell'\in \{-1,0,1\}} \exp(\alpha_{\ell'}^j (\xi_{ij}^{(1)} + \xi_{ij}^{(2)}))} =\\ =\frac{\exp(\alpha_{\ell}^j \xi_{ij}^{(1)}) \exp(\alpha_{\ell}^j \xi_{ij}^{(2)})}{\sum_{\ell'\in \{-1,0,1\}} \exp(\alpha_{\ell'}^j \xi_{ij}^{(1)})\exp(\alpha_{\ell'}^j \xi_{ij}^{(2)})} = 
	\frac{p_{\ell}^{(1)} p_{\ell}^{(2)}}
	{\sum_{\ell' \in \{-1,0,1\}} p_{\ell'}^{(1)} p_{\ell'}^{(2)}} = p_{\ell}^{\text{joint}}.
	\end{multline*}
	
	Thus the function $F_\ell^j(\xi_{ij})$ coincides with the Bayesian fusion of $F_\ell^j(\xi_{ij}^{(1)})$ and $F_\ell^j(\xi_{ij}^{(2)})$. 
	Under the assumption that the two agents' observations are conditionally independent given the true state of $A$, this normalised product coincides with the Bayesian posterior update (under a uniform prior) combining both pieces of evidence.
\end{proof}

The exponential parametrisation $F_{\ell}^j(\xi) \propto \exp(\alpha_{\ell}^j \xi)$ 
	is natural in this context, as it makes the space $\mathcal{B}^m_{\mathbf{F}}(n)$ a linear exponential family. In that case, the map $\xi \mapsto (F_{-1}^j(\xi), F_0^j(\xi), F_1^j(\xi))$ is the soft‑max transform, widely used in statistical modelling and machine learning. The additive combination of weights then corresponds to 
	multiplying independent probabilities, which is the standard rule for combining independent evidence.
	
The previous theorem shows that if the probability maps are of the exponential (softmax) form, then the addition of weights corresponds exactly to Bayesian fusion. Essentially, it is a well-known result rewritten in the article's framework. Further discussion on this matter can be found in Appendix \ref{app:fusions}.

Let us describe the underlying deterministic structure for a PDNF. Consider a PDNF with a fixed but unknown weight matrix $\Xi$. The set of venjunctions that can actually occur (i.e., have positive probability) is some subset $\mathcal{V}_+ \subseteq \mathcal{V}(m,n)$; let $M = |\mathcal{V}_+| \le 3^{nm}$.  
Assume that from physical or design considerations, we know a uniform lower bound on the probabilities of all occurring venjunctions:
\[
\mu_\Xi(\{\langle v\rangle \}) \ge p_{\min} > 0 \qquad \text{for every } \langle v\rangle \in \mathcal{V}_+.
\]

We perform independent draws of venjunctions according to $\mu_\Xi$.  
A classical result from the coupon collector problem tells us that to see every element of a set of size $M$ at least once with probability at least $1-\delta$, it suffices to take
\[
N \ge \frac{1}{p_{\min}} \ln\frac{M}{\delta}
\]
observations. So, we can formulate the following statement 

\begin{prop}\label{prop:regexp}
	Let $\mathcal{V}_+ \subseteq \mathcal{V}(m,n)$ denote the set of all venjunctions with positive probability under $\mu_\Xi$, with $M = |\mathcal{V}_+| \le 3^{nm}$, and suppose a uniform lower bound $p_{\min} > 0$ satisfying $\mu_\Xi(\{\langle v\rangle\}) \ge p_{\min}$ for all $\langle v\rangle \in \mathcal{V}_+$ is known a priori. Then
	\begin{equation}\label{bound:coupon}
	N \ge \frac{1}{p_{\min}} \ln\frac{M}{\delta}
	\end{equation}
	independent observations suffice to observe every venjunction in $\mathcal{V}_+$ at least once with probability at least $1-\delta$. When $M$ is unknown, the worst-case bound is
	$$N \ge \frac{1}{p_{\min}} \ln\frac{3^{nm}}{\delta}.$$
\end{prop}

\begin{proof}
	The probability that a fixed venjunction $v\in\mathcal{V}_+$ is not observed in $N$ trials is at most $(1-p_{\min})^N \le e^{-p_{\min}N}$.  
	By the union bound, the chance that at least one of the $M$ venjunctions remains unseen does not exceed $M e^{-p_{\min}N} \le 3^{nm} e^{-p_{\min}N}$.  
	Setting $M e^{-p_{\min}N} \le \delta$ yields the bound (\ref{bound:coupon}), and using $M \le 3^{nm}$ gives the worst-case bound.
\end{proof}

Thus, if a positive lower bound $p_{\min}$ is known a priori, we can guarantee that after $N$ observations (with $N$ chosen as above) we will have seen every possible venjunction with confidence $1-\delta$.  
This is useful, for example, when the PDNF model is used to describe a physical system whose behaviour is known to have a certain minimum probability for each of its possible states.

Suppose that after a sufficiently long sequence of independent observations, we have collected the complete set \(\mathcal{V}_+ = \{\langle v_1\rangle,\dots,\langle v_M\rangle\}\) of all venjunctions that can actually occur.  
Each $\langle v_r\rangle$ is a specific deterministic venjunction, i.e., a sequence of literals (each either \(x_j\), \(\overline{x}_j\), or \(\varepsilon\)).  
We now wish to represent this set as a single deterministic venjunction that also involves a collection of hidden variables \(h_1,\dots,h_l\) (the same for all time steps).  
These hidden variables are abstract placeholders; they may appear as themselves, negated, or be absent in each conjunction, exactly like the ordinary variables \(x_j\).

A natural way to encode the set \(\mathcal{V}_+\) is to introduce one hidden variable combination for each observed venjunction.  
If the hidden variables can take values from a finite set (e.g., each \(h_s\) may be present, negated, or absent), then we can write
\[
\mathcal{V}_+ = \left\{\mathop{\vcenter{\hbox{\scalebox{1.5}{$\angle$}}}}_{t=n}^{1} \left( \bigwedge_{j=1}^{m} x_j^{a^r_{tj}} \;\wedge\; \bigwedge_{s=1}^{l} h_s^{b^r_{ts}} \right)\right\}_{r=1}^{M},
\]
where the exponents indicate the form (presence, negation, or absence) of the corresponding variable in the \(t\)-th conjunction in the \(r\)-th venjunction.  
The coefficients \(a_{tj}^r\) are uniquely determined by the observed venjunction \(\langle v_r\rangle \), while the coefficients \(b_{ts}^r\) can be chosen arbitrarily as long as they are distinct for different \(r\) (so that each conjunction corresponds to a unique combination of hidden variables).  
For instance, if we let \(l = \lceil \log_3 M \rceil\) and interpret each combination of \(b_{ts}^r\) as a base‑3 representation of the index \(r\), we obtain a valid encoding.

Thus, the existence of such a representation is always guaranteed, but it is by no means unique.  
Without additional information about the physical meaning of the hidden variables, any bijection between \(\mathcal{V}_+\) and the set of hidden‑variable configurations yields a deterministic description.  
In practice, if the hidden variables are known to represent specific unobserved factors, their influence on the observations may impose constraints that select a particular representation. Otherwise, the constructed deterministic venjunction merely serves as a proof of concept that the probabilistic PDNF can be expanded into a deterministic formula with extra variables.

So, the weights \(\xi_{ij}\) encode the probabilities of the hidden states, and the algebraic operations allow us to combine evidence from multiple independent observations, gradually reducing uncertainty. In the limit, when all possible venjunctions have been observed, the PDNF can be replaced by an equivalent deterministic description, albeit one that may require introducing auxiliary variables.

\section{Connection of PDNFs and Markov Processes\label{sec:Markov}}
The weights $\xi_{ij}$ might not be fixed but evolve over time according to some stochastic process. A natural and important case is when each weight follows a random walk. At each time $t>1$, the weight $\xi_{tj} = \eta_{tj}+\xi_{t-1,j}$ is a realisation of a random walk with i.i.d. increments, and the output probabilities $F_\ell^j(\xi_{tj})$ become random variables themselves.

This construction leads to a two‑level stochastic process (see Subsection~\ref{ssec:probmeasure}). The underlying random walk $\{\xi_j(t)\}_{t\ge 1}$ describes the evolution of the physical quantity measured by sensor $j$. 

The resulting sequence of venjunctions $\{\langle v(t)\rangle \}_{t\ge 1}$ is therefore a hidden Markov model with the random walk as the hidden state and the PDNF specifying the emission probabilities.

Using the continuous representation of PDNFs introduced in Subsection~\ref{ss:cont}, we can consider the mapping $\Pi_z:\mathcal{RW} \rightarrow \mathcal{E}^m_{\mathbf{F}}(n)$, where $\mathcal{RW}$ is the set of considered random walks, $\Pi_z (\xi)\mapsto z(\xi; \cdot)\in \mathcal{E}^m_{\mathbf{F}}(n)$, $\xi \in \mathcal{RW}$. The $\Pi_z(\xi)$ is an encoder-valued stochastic process in time $t$, with values of the form $z(\xi(t); x)$. The mean function of $\Pi_z$ is given by the Bochner integral
\[
\mathbb{E}[\Pi_z(\xi)](x) = \mathbb{E}[z(\xi;x)] = \int_{\xi \in \mathcal{RW}} z(\xi;x) \, d\mathbb{P}_\xi,
\]
where $\mathbb{P}_\xi$ is the distribution of $\xi$. For the piecewise constant representation of a PDNF,
\[
(\Pi_z(\xi))(x) = \sum_{t=1}^n\sum_{j=1}^m \xi_{tj} \mathbf{1}_{\Delta_{tj}}(x),
\]
this simplifies to
\[
\mathbb{E}[\Pi_z(\xi)](x) = \sum_{t=1}^n\sum_{j=1}^m \mathbb{E}[\xi_{tj}] \,\mathbf{1}_{\Delta_{tj}}(x) = \sum_{t=1}^n\sum_{j=1}^m (t\mathbb{E}\eta+\mathbb{E}\xi_{0,j})\mathbf{1}_{\Delta_{tj}}(x).
\]
Thus, the mean function is again piecewise constant, with step heights given by the expected values of the weights at time $t$. If $\xi_{0,j}$ is deterministic, $\mathbb{E}[\xi_{tj}] = \xi_{0,j} + t\mu$ where $\mu = \mathbb{E}[\eta]$.

The variance of the encoder $\Pi_z(\xi)$ at a point $x$ is defined as
\[
\operatorname{Var}(\Pi_z(\xi))(x) = \mathbb{E}\bigl[(\Pi_z(\xi)(x) - \mathbb{E}[\Pi_z(\xi)(x)])^2\bigr].
\]

Using the piecewise constant representation
\[
(\Pi_z(\xi))(x) = \sum_{t=1}^n\sum_{j=1}^m \xi_{tj} \mathbf{1}_{\Delta_{tj}}(x),
\]
and noting that the indicator functions have disjoint supports for different pairs $(t,j)$, we obtain
\[
\operatorname{Var}(\Pi_z(\xi))(x) = \sum_{t=1}^n\sum_{j=1}^m \operatorname{Var}(\xi_{tj}) \,\mathbf{1}_{\Delta_{tj}}(x).
\]

If $\xi_{tj}$ follows a random walk $\xi_{tj} = \xi_{0j} + \sum_{s=1}^t \eta_s$ with i.i.d. increments $\eta_s$ having variance $\sigma^2 = \operatorname{Var}(\eta_s)$, and initial values $\xi_{0,j}$ (possibly random with variance $\operatorname{Var}(\xi_{0,j})$), then
\[
\operatorname{Var}(\xi_{tj}) = \operatorname{Var}(\xi_{0,j}) + t\sigma^2.
\]
Substituting, we obtain that
\[
\operatorname{Var}(\Pi_z(\xi))(x) = \sum_{t=1}^n\sum_{j=1}^m \bigl(\operatorname{Var}(\xi_{0,j}) + t\sigma^2\bigr) \mathbf{1}_{\Delta_{tj}}(x).
\]

For a symmetric random walk ($\mu=0$), the mean remains constant while the variance grows, reflecting increasing uncertainty over time. For a random walk with drift, both the mean and the variance grow linearly.

Moderately applied examples for this theory will be given in Appendix~\ref{app:samples}. We can also prove the existence of the mean value for a PDNF encoder in the general case.

\begin{thm}
	Let \(\mathcal{RW}\) be a set of trajectories of a random walk with a probability measure $\mathbb{P}_\xi$ and finite first moments $\mathbb{E}[|\xi_{tj}|] < \infty$ for all $t,j$.
	
	Let $\overline{\mathcal{E}^m_{\mathbf{F}}(n)} \subseteq L_1([0,nm])$ denote the closure of $\mathcal{E}^m_{\mathbf{F}}(n)$ in the Banach space $L_1([0,nm])$ equipped with the norm $\|f\|_1 = \int_{0}^{nm} |f(x)|\,dx$.
	
Assume that $\Pi_z(\xi)$ depends on $\xi$ only through its restriction $(\xi_{tj})_{t=\overline{1,n}, j=\overline{1,m}} \in \mathbb{R}^{n \times m}$, that this dependence is continuous with respect to the Euclidean topology on $\mathbb{R}^{n \times m}$ and the $L_1$-norm on $\overline{\mathcal{E}_{\mathbf{F}}^{m}(n)}$, and that $\Pi_z$ satisfies the linear growth bound
	\begin{equation}\label{eq:RWass}
		\|\Pi_z(\xi)\|_1 \le C \sum_{t=1}^n \sum_{j=1}^m |\xi_{tj}|
	\end{equation}
	for some constant \(C>0\).
	
	Then the Bochner integral
	\[
	\mathbb{E}[\Pi_z] = \int_{\mathcal{RW}} \Pi_z(\xi) \, d\mathbb{P}_\xi \in \overline{\mathcal{E}^m_{\mathbf{F}}(n)}
	\]
	exists.
\end{thm}

\begin{proof}	
	Let $(\Omega,\mathcal{F})$ be a measurable space and $X$ a Banach space. A function $s:\Omega\to X$ is called \emph{simple} if it can be written as a finite linear combination of indicator functions of measurable sets:
	\[
	s(\omega)=\sum_{k=1}^{n} \mathbf{1}_{A_k}(\omega)\, x_k,
	\]
	where $A_k\in\mathcal{F}$ are pairwise disjoint measurable sets and $x_k\in X$ are fixed vectors.
	
	We need to prove strong measurability of \(\Pi_z\), that is a sequence of simple functions \(\Pi_k: \mathcal{RW} \to \mathcal{E}^m_{\mathbf{F}}(n)\) such that \(\|\Pi_k(\xi) - \Pi_z(\xi)\|_1 \to 0\) for \(\mathbb{P}_\xi\)-almost every \(\xi\). 
	
	Also, it is necessary to prove the integrability of the norm: 
	\[\int_{\mathcal{RW}} \|\Pi_z(\xi)\|_1 \, d\mathbb{P}_\xi < \infty.\]

\paragraph{Strong measurability.}
The map $\xi \mapsto \Pi_z(\xi)$ is a composition of the random walk trajectory $\xi$ (which is a measurable function) with the construction of the encoder $z(\xi;\cdot)$ from Subsection~\ref{ss:cont}. This encoder depends continuously on $\xi$, so it is Borel measurable.  Moreover, $\mathcal{E}^m_{\mathbf{F}}(n)\subset L_1$ is separable, so the image of $\Pi_z$ is separably valued. Hence $\Pi_z$ is strongly measurable by virtue of Pettis theorem \cite[Chapter II, Theorem 2]{diestel1977vector}: it is the pointwise limit of simple functions (for instance, by approximating the continuous map by piecewise constant functions in the Banach space).

\paragraph{Integrability of the norm.}
From the assumption (\ref{eq:RWass}), we have
\begin{equation}
\label{eq:sumRW}
\int_{\mathcal{RW}} \|\Pi_z(\xi)\|_1 \, d\mathbb{P}_\xi \le C \sum_{t=1}^n \sum_{j=1}^m \mathbb{E}[|\xi_{tj}|].
\end{equation}
If the random walk has finite first moments (e.g., bounded increments), then \(\mathbb{E}[|\xi_{tj}|] < \infty\) for all \(t,j\), and the sum (\ref{eq:sumRW}) is finite. 
\end{proof}

Note that for piecewise constant encoders on the fixed partition $\{\Delta_{tj}\}$, the space $\mathcal{E}^m_{\mathbf{F}}(n)$ is finite-dimensional and hence already closed and Banach. The estimate (\ref{eq:RWass}) holds for the piecewise constant representation, and more generally whenever the dependence on \(\xi\) is Lipschitz in the \(L_1\)-norm.

The mean encoder function $\mathbb{E}[\Pi_z(\xi)](x)$ obtained via the Bochner integral provides, through its values on each subinterval $\Delta_{tj}$, the expected weight $\mathbb{E}[\xi_{tj}]$ at time $t$ for sensor $j$. Note that, because the probability maps $F_\ell^j$ are generally nonlinear, the expected weight determines the average output probabilities $\mathbb{E}[F_\ell^j(\xi_{tj})]$ only up to a first-order approximation (in general, $\mathbb{E}[F_\ell^j(\xi_{tj})] \neq F_\ell^j(\mathbb{E}[\xi_{tj}])$). If the weight distribution is multimodal (e.g., concentrated equally on large positive and large negative values), the mean $\mathbb{E}[\xi_{tj}]$ may lie near zero even though the system rarely visits the indeterminate state $\varepsilon$. Consequently, $\mathbb{E}[\Pi_z(\xi)]$ should be interpreted strictly as a first-moment summary in the parameter space (a ``central'' weight profile), which is most informative for unimodal distributions with moderate variance (such as random walks over short horizons). In such cases, if the expected weight is significantly positive, the system tends to output $x_j$; if significantly negative, it tends to output $\overline{x}_j$; if the weight is near zero ( (under generative mappings satisfying $F_0^j(0)=1$), the output is most likely $\varepsilon$. Thus, the mean function can be used for decision-making: for example, one may trigger an alarm or take preventive action when the expected weight exceeds a certain threshold, indicating a high probability of a critical event.  

To mitigate the cancellation of opposite tendencies when averaging multimodal or symmetric weight distributions, one can consider the positive and negative parts of the encoder separately. Let $z_+(\xi;\cdot) = \max(z,0)$ and $z_-(\xi;\cdot) = \min(z,0)$, so that $z(\xi;\cdot) = z_+(\xi;\cdot) + z_-(\xi;\cdot)$. The expectations $\mathbb{E}[z_+(\xi;\cdot)]$ and $\mathbb{E}[z_-(\xi;\cdot)]$ (computed as Bochner integrals of the corresponding functions) describe the average positive and negative weight magnitudes, respectively, while their difference or ratio serves as a measure of PDNF asymmetry.

In terms of the hypothesis-generation viewpoint of Section~\ref{sec:regexp}, this expected encoder corresponds to a ``central'' formula representing the average parameter profile across the family of possible venjunctions.

\section{Conclusion}

In this paper, we introduced probabilistic disjunctive normal forms (PDNFs) as a formalism that combines Boolean logic and probabilistic modelling. Our main contribution is a novel continuous representation: each PDNF is encoded by a piecewise-continuous function on a partitioned interval, where integrals over subintervals determine the probabilities of PDNF literals. This perspective turns the set of PDNFs into a Banach space, allowing us to treat random PDNFs as stochastic processes taking values in function spaces. We connected PDNFs with weights evolving as random walks to hidden Markov models, and computed explicit expressions for the mean and variance of the corresponding encoder process.

Also, we establish several minor results:
\begin{itemize}[label=$\triangleright$]
	\item We proved that each PDNF induces a probability measure on the set of temporal logic statements (venjunctions), and that the expectation of a random encoder exists as a Bochner integral under mild conditions.
	\item We showed that the addition of PDNFs corresponds, under an exponential parametrisation, to Bayesian fusion of independent evidence, and we characterised the exponential family as the only one with this property.
	\item We obtained quantitative bounds for the number of observations needed to identify all possible venjunctions.
	\item We also discussed applications to automata identification and to generating families of plausible logical statements.
\end{itemize}
The proposed formalism bridges discrete logic, probability theory, and functional analysis. Future work may explore general infinite-dimensional extensions and deductive inference rules, the use of more general function spaces, and applications in statistical learning and verification. In particular, from a practical standpoint, the algebraic structure of PDNFs makes them well-suited as a lightweight neurosymbolic layer for modern machine learning architectures. Potential applications include aggregating noisy multimodal sensor data with text-derived facts, performing linear probing of large language models to extract explicit temporal logic representations, and designing smooth semantic loss functions for logical consistency.

\bibliography{ESU1}	

@inproceedings{green2007provenance,
	author    = {Green, Todd J. and Karvounarakis, Grigoris and Tannen, Val},
	title     = {Provenance semirings},
	booktitle = {Proceedings of the Twenty-Sixth ACM SIGMOD-SIGACT-SIGART Symposium on Principles of Database Systems (PODS)},
	pages     = {31--40},
	year      = {2007},
	doi       = {10.1145/1265530.1265535}
}

@book{suciu2011probabilistic,
	author    = {Suciu, Dan and Olteanu, Dan and R{\'e}, Christopher and Koch, Christoph},
	title     = {Probabilistic Databases},
	publisher = {Morgan \& Claypool Publishers},
	year      = {2011},
	doi       = {10.2200/S00362ED1V01Y201105DTM016}
}

@article{zadeh1965fuzzy,
	 author={Zadeh, Lotfi A.},
	 title={Fuzzy sets},
	 journal={Information and Control},
	 volume={8},
	 number={3},
	 pages={338--353}, 
	 year={1965},
	 doi = {10.1016/S0019-9958(65)90241-X},
	 url = {https://www.sciencedirect.com/science/article/pii/S001999586590241X}
 }

@article{kuznetsov2026modelling,
	author = {Kuznetsov, Alexander and Sytnikov, Dmytro and Sultonov, Shoxrux},
	title = {Modelling the process of agent cognition under conditions of incomplete or incorrect information},
	journal = {Journal of Information and Optimization Sciences},
	year = {2026},
	doi = {10.47974/JIOS-2229},
}

@article{KuznetsovLemtyuzhnikova2026,
	author    = {Kuznetsov, A. V. and Lemtyuzhnikova, D. V.},
	title     = {Knowledge graph metrics-based formal theory of argumentation},
	journal   = {Upravlenie Bol'shimi Sistemami},
	year      = {2026},
	number    = {122},
	pages     = {52--101},
	note      = {(in Russian)},
	url		  = {https://www.mathnet.ru/rus/ubs/v122/p52},
	language  = {russian}
}

@article{KarnikMendel2001,
	title = {Centroid of a type-2 fuzzy set},
	journal = {Information Sciences},
	volume = {132},
	number = {1},
	pages = {195-220},
	year = {2001},
	issn = {0020-0255},
	doi = {https://doi.org/10.1016/S0020-0255(01)00069-X},
	url = {https://www.sciencedirect.com/science/article/pii/S002002550100069X},
	author = {Nilesh N. Karnik and Jerry M. Mendel}
}

@misc{mertens2025parametric,
	title={Compositional Reasoning for Parametric Probabilistic Automata}, 
	author={Hannah Mertens and Tim Quatmann and Joost-Pieter Katoen},
	year={2025},
	eprint={2506.08525},
	archivePrefix={arXiv},
	primaryClass={cs.LO},
	url={https://arxiv.org/abs/2506.08525}, 
}

@inproceedings{derkinderen2025compiling,
	title     = {Circuit-Aware d-DNNF Compilation},
	author    = {Derkinderen, Vincent and Lagniez, Jean-Marie},
	booktitle = {Proceedings of the Thirty-Fourth International Joint Conference on
	Artificial Intelligence, {IJCAI-25}},
	publisher = {International Joint Conferences on Artificial Intelligence Organization},
	editor    = {James Kwok},
	pages     = {4454--4462},
	year      = {2025},
	month     = {8},
	note      = {Main Track},
	doi       = {10.24963/ijcai.2025/496},
	url       = {https://doi.org/10.24963/ijcai.2025/496},
}

@inproceedings{kwiatkowska2025probabilistic,
	author = {Jamroga, Wojciech and Kwiatkowska, Marta and Penczek, Wojciech and Petrucci, Laure and Sidoruk, Teofil},
	title = {Probabilistic Timed ATL},
	year = {2025},
	isbn = {9798400714269},
	publisher = {International Foundation for Autonomous Agents and Multiagent Systems},
	address = {Richland, SC},
	booktitle = {Proceedings of the 24th International Conference on Autonomous Agents and Multiagent Systems},
	pages = {1051--1059},
	numpages = {9},
	location = {Detroit, MI, USA},
	series = {AAMAS '25}
}

@book{diestel1977vector,
	title={Vector Measures},
	author={Diestel, Joseph and Uhl, Jerry J.},
	series={Mathematical Surveys},
	volume={15},
	year={1977},
	publisher={American Mathematical Society},
	address={Providence, RI},
	note={}
}

@book{vasyukevich2011digital,
	author    = {Vadim Vasyukevich},
	title     = {Asynchronous Operators of Sequential Logic: Venjunction \& Sequention. Digital Circuit Analysis and Design},
	year      = {2011},
	publisher = {Springer Berlin, Heidelberg},
	series    = {Lecture Notes in Electrical Engineering},
	volume    = {101},
	isbn      = {978-3-642-21610-7},
	doi       = {10.1007/978-3-642-21611-4},
}

@article{shmulevich2002probabilistic,
	author    = {Ilya Shmulevich and Edward R. Dougherty and Seungchan Kim and Wei Zhang},
	title     = {Probabilistic Boolean Networks: a rule-based uncertainty model for gene regulatory networks},
	journal   = {Bioinformatics},
	year      = {2002},
	volume    = {18},
	number    = {2},
	pages     = {261--274},
	month     = {feb},
	doi       = {10.1093/bioinformatics/18.2.261},
	pmid      = {11847074}
}

@book{shafer1976mathematical,
	author    = {Glenn Shafer},
	title     = {A Mathematical Theory of Evidence},
	year      = {1976},
	publisher = {Princeton University Press},
	address   = {Princeton, NJ},
	url={https://www.jstor.org/stable/j.ctv10vm1qb}
}

@article{richardson2006markov,
	author    = {Matthew Richardson and Pedro Domingos},
	title     = {Markov Logic Networks},
	journal   = {Machine Learning},
	year      = {2006},
	volume    = {62},
	number    = {1-2},
	pages     = {107--136},
	month     = {feb},
	doi       = {10.1007/s10994-006-5833-1},
	issn      = {1573-0565},
	publisher = {Springer}
}

@inbook{Drossel2008,
	author = {Drossel, Barbara},
	publisher = {John Wiley \& Sons, Ltd},
	isbn = {9783527626359},
	title = {Random Boolean Networks},
	booktitle = {Reviews of Nonlinear Dynamics and Complexity},
	chapter = {3},
	pages = {69-110},
	doi = {https://doi.org/10.1002/9783527626359.ch3},
	url = {https://onlinelibrary.wiley.com/doi/abs/10.1002/9783527626359.ch3},
	eprint = {https://onlinelibrary.wiley.com/doi/pdf/10.1002/9783527626359.ch3},
	year = {2008}
}

@article{Kuznetsov2025PDNF,
	author  = {Kuznetsov, A.V.},
	title   = {On statement-valued random variables},
	journal = {Boletín de la Sociedad Matemática Mexicana},
	year    = {2026},
	volume  = {32},
	number  = {8},
	pages   = {8},
	doi     = {10.1007/s40590-025-00840-7}
}

@article{SERGEEV2021134,
	title = {Identification of Integrated Rating Mechanisms As An Approach To Discrete Data Analysis},
	journal = {IFAC-PapersOnLine},
	volume = {54},
	number = {13},
	pages = {134-139},
	year = {2021},
	note = {20th IFAC Conference on Technology, Culture, and International Stability TECIS 2021},
	issn = {2405-8963},
	doi = {https://doi.org/10.1016/j.ifacol.2021.10.433},
	url = {https://www.sciencedirect.com/science/article/pii/S240589632101870X},
	author = {V.A. Sergeev and N.A. Korgin}
}

@article{Wharton1974,
	author = {R. M. Wharton},
	title = {Approximate Language Identification},
	journal = {Information and Control},
	volume = {26},
	pages = {236--255},
	year = {1974},
	doi = {10.1016/S0019-9958(74)91369-2}, 
	url = {https://www.sciencedirect.com/science/article/pii/S0019995874913692}
}

@article{NILSSON198671,
	title = {Probabilistic logic},
	journal = {Artificial Intelligence},
	volume = {28},
	number = {1},
	pages = {71-87},
	year = {1986},
	issn = {0004-3702},
	doi = {https://doi.org/10.1016/0004-3702(86)90031-7},
	url = {https://www.sciencedirect.com/science/article/pii/0004370286900317},
	author = {Nils J. Nilsson}
}

@article{Schumann04052018,
	author = {Andrew Schumann and Alexander V. Kuznetsov},
	title = {Foundations of mathematics under neuroscience conditions of lateral inhibition and lateral activation},
	journal = {International Journal of Parallel, Emergent and Distributed Systems},
	volume = {33},
	number = {3},
	pages = {237--256},
	year = {2018},
	publisher = {Taylor \& Francis},
	doi = {10.1080/17445760.2018.1439490},
	URL = { 
	https://doi.org/10.1080/17445760.2018.1439490   
	},
	eprint = { 
	https://doi.org/10.1080/17445760.2018.1439490
	}
}

@article{sitnikov2002method,
	title = {A Method For Knowledge Representation And Discovery Based On Composing And Manipulating Logical Equations},
	author = {Sitnikov, D. E. and D'Cruz, B. and Sitnikova, P. E.},
	journal = {WIT Transactions on Information and Communication Technologies},
	volume = {28},
	pages = {},
	year = {2002},
	doi = {10.2495/DATA020031},
	copyright = {WIT Press}
}

\appendix
\section{Discussion on the PDNFs Algebra\label{app:algebra}}
Consider a PDNF as representing a series of observations (e.g., sensor readings over time) for features $i\in \{1, \dots, m\}$. Let the exponent $\xi_i$ in a PDNF $Z_1$ be interpreted as a cumulative weight or parameter summarising the influence of feature $i$ in the series $Z_1$, and $\eta_i$ the corresponding weight of feature $i$ in series $Z_2$. If we combine the evidence from series $Z_1$ and $Z_2$ -- assuming they pertain to the same observational framework -- the combined influence for each feature is naturally modelled by the sum $\xi_i+\eta_i$, which corresponds exactly to the addition defined for elementary conjunctions (and extended to PDNFs via (\ref{eqsum})). Consequently, the norm (\ref{eq:norm1}) provides a measure of the total absolute weight of a PDNF. A small distance between two PDNFs, measured by $\|Z_1-Z_2\|_1$, indicates that the two series assign similar cumulative weights to the features. 

Consider a fixed position \((i,j)\) and a sequence of independent PDNFs \(Z_1,Z_2,\ldots\) with weights \(\xi_{ij}^{(k)}\) at that position. Assume that all these weights have the same sign, e.g., \(\xi_{ij}^{(k)}\ge 0\) for every \(k\) (or all \(\le 0\)).  
Form the cumulative PDNF \(Z^{(N)} = Z_1 + \cdots + Z_N\); its weight at position \((i,j)\) is \(\xi_{ij}^{(N)} = \sum_{k=1}^N \xi_{ij}^{(k)}\).

If the weights are non-negative and there exists a constant $c>0$ such that infinitely many of them satisfy $\xi_{ij}^{(k)}\ge c$, then $\xi_{ij}^{(N)}\to +\infty$ as $N\to\infty$. By the second Borel--Cantelli lemma, this holds almost surely, for instance, when the independent observations are identically distributed with $\mathbb{P}\{\xi_{ij}^{(k)} \ge c\} > 0$ (or, more generally, whenever $\sum_{k=1}^{\infty} \mathbb{P}\{\xi_{ij}^{(k)} \ge c\} = \infty$).

By the monotonicity properties of the maps $F_\ell^j$ (Definition~\ref{def:def01}) and the natural asymptotic assumption that $F_1^j(\xi)\to 1$ and $F_0^j(\xi),F_{-1}^j(\xi)\to 0$ as $\xi\to+\infty$, we obtain
\[
F_1^j(\xi_{ij}^{(N)}) \longrightarrow 1,\qquad 
F_0^j(\xi_{ij}^{(N)}) \longrightarrow 0,\qquad 
F_{-1}^j(\xi_{ij}^{(N)}) \longrightarrow 0
\]
as $N\to\infty$. Hence, for sufficiently large $N$, the probability of observing $x_j$ at this position approaches $1$ (and if $F_1^j$ saturates to $1$ above a finite threshold, realisations at this position eventually coincide with $x_j$ almost surely).

Similarly, if all weights are non-positive and $\xi_{ij}^{(k)}\le -c$ for infinitely many $k$, then $\xi_{ij}^{(N)}\to -\infty$ and, under the analogous asymptotic behaviour of $F_{-1}^j$, we get
\[
F_{-1}^j(\xi_{ij}^{(N)}) \longrightarrow 1,\qquad 
F_0^j(\xi_{ij}^{(N)}) \longrightarrow 0,\qquad 
F_1^j(\xi_{ij}^{(N)}) \longrightarrow 0,
\]
so the probability of observing $\overline{x}_j$ tends to $1$.

This result formalises the intuition that repeated evidence of the same type (e.g., a sensor consistently triggered) gradually eliminates uncertainty and drives the system toward a deterministic outcome. 

\section{Discussion on the PDNF Structure\label{app:struct}} 
Essentially, a disjunctive normal form of $m$ variables can be represented as a tree graph (see Fig.~\ref{fig:figa}). The leaves of this graph correspond to variables, the root vertex represents the disjunction symbol $\vee$, the root's children are conjunction symbols $\wedge$, and the children of each conjunction symbol are either negation symbols $\neg$ or identity symbols $\mathrm{id}$ (where $\mathrm{id}(x)=x$). The edges of the tree are unmarked; thus, children of the same parent can be placed in any order without affecting the logical meaning.

A PDNF modifies this structure in two key ways (see Fig.~\ref{fig:figb}): first, negation or identity symbols are replaced by real-valued weights $\xi_i$; second, and more importantly, the edges connecting the root to the conjunction symbols are marked (e.g., with time indices). These marks record, for instance, the observation moment of each conjunction, making the order of the conjunctions fixed and semantically significant.

\begin{figure}[htbp]
	\centering
	\begin{subfigure}{0.45\textwidth}
		\centering
		\begin{tikzpicture}[
			every node/.style={draw, rectangle, minimum size=6mm, inner sep=2pt, font=\small},
			level 1/.style={sibling distance=2.5cm, level distance=1.2cm},
			level 2/.style={sibling distance=1.2cm, level distance=1.2cm},
			level 3/.style={level distance=1.2cm},
			edge from parent/.style={draw, gray!80, thick}
			]
			\node {$\vee$}
			child { node {$\wedge$}
				child { node {$\mathrm{id}$}
					child { node {$x$} }
				}
				child { node {$\neg$}
					child { node {$y$} }
				}
			}
			child { node {$\wedge$}
				child { node {$\mathrm{id}$}
					child { node {$x$} }
				}
				child { node {$\mathrm{id}$}
					child { node {$y$} }
				}
			};
		\end{tikzpicture}
		\caption{}
		\label{fig:figa}
	\end{subfigure}
	\hfill
	\begin{subfigure}{0.45\textwidth}
		\centering
		\begin{tikzpicture}[
			every node/.style={draw, rectangle, minimum size=6mm, inner sep=2pt, font=\small},
			level 1/.style={sibling distance=2.5cm, level distance=1.2cm},
			level 2/.style={sibling distance=1.2cm, level distance=1.2cm},
			level 3/.style={level distance=1.2cm},
			edge from parent/.style={draw, gray!80, thick}
			]
			\node {$\sqcup$}
			child { node {$\wedge$}
				child { node {$\xi_1$}
					child { node {$x$} }
				}
				child { node {$\eta_1$}
					child { node {$y$} }
				}
				edge from parent node[left, draw=none, font=\scriptsize, xshift=-2pt] {$t_1$}
			}
			child { node {$\wedge$}
				child { node {$\xi_2$}
					child { node {$x$} }
				}
				child { node {$\eta_2$}
					child { node {$y$} }
				}
				edge from parent node[right, draw=none, font=\scriptsize, xshift=2pt] {$t_2$}
			};
		\end{tikzpicture}
		\caption{}
		\label{fig:figb}
	\end{subfigure}
	\caption{Comparison of the classical DNF $x\wedge \overline{y}\vee x\wedge y$ (a) and the PDNF $x^{\xi_1}\wedge y^{\eta_1}\sqcup x^{\xi_2}\wedge y^{\eta_2}$ (b)}
	\label{fig:trees}
\end{figure}

\section{Discussion on the PDNFs Measure\label{app:prob}} 
\begin{prop}\label{prop:measure}
	
	The function $\mu_{\Xi}:2^{\mathcal{V}(m,n)} \rightarrow \mathbb{R}$ extended additively to arbitrary subsets $A \subseteq \mathcal{V}(m,n)$ as
	\[
	\mu_{\Xi}(A) = \sum_{\langle v\rangle \in A} \mu_{\Xi}(\{\langle v\rangle \}).
	\] 
	is a probability measure on the measurable space $(\mathcal{V}(m,n), 2^{\mathcal{V}(m,n)})$.
\end{prop}

\begin{proof}
	We should verify the two defining properties of a probability measure:
	\begin{enumerate}
		\item Non-negativity and normalization: $0 \leq \mu_{\Xi}(A) \leq 1$ for every $A \subseteq \mathcal{V}(m,n)$, and $\mu_{\Xi}(\mathcal{V}(m,n)) = 1$.
		\item Countable additivity (which reduces to finite additivity because $\mathcal{V}(m,n)$ is finite): for any pairwise disjoint sets $A_1, A_2, \dots \subseteq \mathcal{V}(m,n)$
		\[
		\mu_{\Xi}\!\left(\bigcup_{k=1}^{\infty} A_k\right) = \sum_{k=1}^{\infty} \mu_{\Xi}(A_k).
		\]
	\end{enumerate}
	
	\noindent\textbf{Non-negativity.} By Definition~\ref{def:def01}, the functions $F_{-1}^j$, $F_{0}^j$, $F_{1}^j$, $j=\overline{1,m}$, take values in $[0,1]$. Hence each factor in the product defining $\mu_{\Xi}(\{\langle v\rangle\})$ is non‑negative, and consequently $\mu_{\Xi}(\{\langle v\rangle\}) \geq 0$ for every $\langle v\rangle \in \mathcal{V}(m,n)$. For any subset $A \subseteq \mathcal{V}(m,n)$, the measure $\mu_{\Xi}(A)$ is a sum of non‑negative numbers, so $\mu_{\Xi}(A) \geq 0$.
	
	\noindent\textbf{Normalization.} We show that the total measure of the whole space equals one:
	\[
	\mu_{\Xi}(\mathcal{V}(m,n)) = \sum_{\langle v\rangle \in \mathcal{V}(m,n)} \mu_{\Xi}(\{\langle v\rangle\}) = 1.
	\]
	
	The set $\mathcal{V}(m,n)$ corresponds bijectively to all choices of $v_{ij} \in \mathcal{C}=\{\overline{x}_j, \varepsilon, x_j\}$ for each pair $(i,j)$ with $i=\overline{1,n}$, $j=\overline{1,m}$. Under the assumption that the choices at different positions $(i,j)$ are independent, we can exchange the product and the sums:
	\begin{align*}
		\sum_{\langle v\rangle \in \mathcal{V}(m,n)} \mu_{\Xi}(\{\langle v\rangle\})
		&= \sum_{v_{11}\in \mathcal{C}} \cdots \sum_{v_{nm}\in \mathcal{C}} \prod_{i=1}^n \prod_{j=1}^m
		\begin{cases}
			F_{-1}^j(\xi_{ij}) &  \text{if } v_{ij} = \overline{x}_j,\\
			F_{0}^j(\xi_{ij}) &  \text{if } v_{ij} = \varepsilon,\\
			F_{1}^j(\xi_{ij}) &  \text{if } v_{ij} = x_j,
		\end{cases}= \\
		&= \prod_{i=1}^n \prod_{j=1}^m
		\Bigl( F_{-1}^j(\xi_{ij}) + F_{0}^j(\xi_{ij}) + F_{1}^j(\xi_{ij}) \Bigr).
	\end{align*}
	By Definition~\ref{def:def01}, for every $j=\overline{1,m}$ and every $x\in \mathbb{R}$ we have
	\[
	F_{-1}^j(x) + F_{0}^j(x) + F_{1}^j(x) = 1.
	\]
	Applying this identity with $x = \xi_{ij}$ gives $F_{-1}^j(\xi_{ij}) + F_{0}^j(\xi_{ij}) + F_{1}^j(\xi_{ij}) = 1$ for each $(i,j)$. Therefore
	\[
	\sum_{\langle v\rangle \in \mathcal{V}(m,n)} \mu_{\Xi}(\{\langle v\rangle\}) = \prod_{i=1}^n \prod_{j=1}^m 1 = 1.
	\]
	
	\noindent\textbf{Additivity.} Since $\mathcal{V}(m,n)$ is finite, it suffices to prove finite additivity. Let $A_1, A_2, \dots, A_K \subseteq \mathcal{V}(m,n)$ be pairwise disjoint. Then
	\[
	\mu_{\Xi}\!\left(\bigcup_{k=1}^K A_k\right)
	= \sum_{\langle v\rangle \in \bigcup_{k=1}^K A_k} \mu_{\Xi}(\{\langle v\rangle\})
	= \sum_{k=1}^K \sum_{\langle v\rangle \in A_k} \mu_{\Xi}(\{\langle v\rangle\})
	= \sum_{k=1}^K \mu_{\Xi}(A_k).
	\]
	
	Both conditions are satisfied, hence $\mu_{\Xi}$ is a probability measure.
\end{proof}

\section{PDNFs and Bayesian Fusions\label{app:fusions}}
\begin{thm}\label{thm:characterization}
	Let $F_\ell: \mathbb{R} \to (0,1)$, $\ell \in \{-1,0,1\}$ be continuous functions such that $\sum_{\ell \in \{-1,0,1\}} F_\ell(\xi)=1$ for all $\xi \in \mathbb{R}$. Suppose that for all $\xi, \eta \in \mathbb{R}$ the following property holds:
	\[
	F_\ell(\xi+\eta) = \frac{F_\ell(\xi) F_\ell(\eta)}{\sum_{\ell' \in \{-1,0,1\}} F_{\ell'}(\xi) F_{\ell'}(\eta)}, \qquad \ell \in \{-1,0,1\}.
	\]
	Then there exist constants $\alpha_\ell \in \mathbb{R}$ such that
	\[
	F_\ell(\xi) = \frac{\exp(\alpha_\ell \xi)}{\sum_{\ell' \in \{-1,0,1\}} \exp(\alpha_{\ell'} \xi)}.
	\]
\end{thm}

\begin{proof}
	Define $g_\ell(\xi) = \ln F_\ell(\xi)$. The hypothesis becomes
	\[
	\exp\bigl(g_\ell(\xi+\eta)\bigr) = \frac{\exp\bigl(g_\ell(\xi)+g_\ell(\eta)\bigr)}{\sum_{\ell'\in \{-1,0,1\}} \exp\bigl(g_{\ell'}(\xi)+g_{\ell'}(\eta)\bigr)}.
	\]
	Taking logarithms yields
	\[
	g_\ell(\xi+\eta) = g_\ell(\xi) + g_\ell(\eta) - \ln\!\left(\sum_{\ell'\in \{-1,0,1\}} \exp\bigl(g_{\ell'}(\xi)+g_{\ell'}(\eta)\bigr)\right).
	\]
	The last term is independent of $\ell$; denote it by $h(\xi,\eta)$. Then for any $\ell$,
	\[
	g_\ell(\xi+\eta) - g_\ell(\xi) - g_\ell(\eta) = h(\xi,\eta).
	\]
	The left-hand side depends on $\ell$, while the right-hand side does not. Hence for any two indices $\ell, \ell'$,
	\[
	g_\ell(\xi+\eta) - g_\ell(\xi) - g_\ell(\eta) = g_{\ell'}(\xi+\eta) - g_{\ell'}(\xi) - g_{\ell'}(\eta).
	\]
	Rearranging,
	\[
	\bigl(g_\ell(\xi+\eta)-g_{\ell'}(\xi+\eta)\bigr) - \bigl(g_\ell(\xi)-g_{\ell'}(\xi)\bigr) - \bigl(g_\ell(\eta)-g_{\ell'}(\eta)\bigr) = 0.
	\]
	Define $\phi_{\ell,\ell'}(\xi) = g_\ell(\xi) - g_{\ell'}(\xi)$. Then the above equation reads
	\[
	\phi_{\ell,\ell'}(\xi+\eta) = \phi_{\ell,\ell'}(\xi) + \phi_{\ell,\ell'}(\eta).
	\]
	Thus $\phi_{\ell,\ell'}$ satisfies Cauchy's functional equation. By continuity of $F_\ell$, the functions $g_\ell$ are continuous, hence $\phi_{\ell,\ell'}$ is continuous. The only continuous solutions to Cauchy's equation are linear: $\phi_{\ell,\ell'}(\xi) = c_{\ell,\ell'} \xi$ for some constant $c_{\ell,\ell'} \in \mathbb{R}$. Moreover, $\phi_{\ell,\ell'} = -\phi_{\ell',\ell}$, so $c_{\ell,\ell'} = -c_{\ell',\ell}$.
	
	Now fix a reference index, say $0$. For each $\ell$, set $\alpha_\ell = c_{\ell,0}$. Then $\phi_{\ell,0}(\xi) = \alpha_\ell \xi$, i.e.,
	\[
	g_\ell(\xi) = g_0(\xi) + \alpha_\ell \xi.
	\]
	
	Using the normalization condition $\sum_{\ell \in \{-1,0,1\}} \exp(g_\ell(\xi)) = 1$, we obtain
	\[
	\sum_{\ell \in \{-1,0,1\}} \exp\bigl(g_0(\xi) + \alpha_\ell \xi\bigr) = \exp(g_0(\xi)) \sum_{\ell} \exp(\alpha_\ell \xi) = 1.
	\]
	Hence $\exp(g_0(\xi)) = \bigl(\sum_{\ell} \exp(\alpha_\ell \xi)\bigr)^{-1}$, and therefore
	\[
	F_\ell(\xi) = \exp(g_\ell(\xi)) = \frac{\exp(\alpha_\ell \xi)}{\sum_{\ell'} \exp(\alpha_{\ell'} \xi)}.
	\]
	This completes the proof.
\end{proof}

The theorem shows that the exponential (softmax) form is the only one for which the addition of weights in the PDNF space corresponds exactly to Bayesian fusion of independent evidence. In practice, other families $\mathbf{F}$ that satisfy the conditions of Definition~\ref{def:def01} may still be used, but the additive combination then approximates Bayesian fusion rather than implementing it exactly. The additive structure still captures the idea that larger weights $\xi_{ij}$ correspond to stronger, more definite predictions.

\section{Discussion on the Probability to Generate a Venjunction Close to the Given One \label{app:close}}
Let us estimate the probability that a randomly generated venjunction $\langle v\rangle$ (according to $\mu_\Xi$) is close to a fixed target venjunction $\langle v'\rangle$ of the same length under a suitable metric. A natural metric on $\mathcal{V}(m,n)$ is a weighted Hamming distance. For two literals $a,b \in \{\overline{x}_j,\varepsilon,x_j\}$ (the index $j$ is the same because positions are aligned), define

\[
w(a,b) =
\begin{cases}
	0, & a = b,\\[2pt]
	1, & (a = \varepsilon,\; b\in\{x_j,\overline{x}_j\}) \text{ or } (b = \varepsilon,\; a\in\{x_j,\overline{x}_j\}),\\[2pt]
	2, & \{a,b\} = \{x_j,\overline{x}_j\}.
\end{cases}
\]
For two venjunctions $\langle v\rangle = v_{n1}\wedge\dots\wedge v_{nm}\angle \dots \angle v_{11}\wedge\dots\wedge v_{1m}$ and $\langle v'\rangle = v'_{n1}\wedge\dots\wedge v'_{nm}\angle \dots \angle v'_{11}\wedge\dots\wedge v'_{1m}$, we set
\[
d(v,v') = \sum_{i=1}^n \sum_{j=1}^m w\bigl(v_{ij},v'_{ij}\bigr).
\]
It is easy to verify that $d$ is a metric (non‑negative, symmetric, satisfies the triangle inequality) and takes integer values from $0$ to $2nm$.

Fix a target venjunction $\langle v'\rangle$ and a weight matrix $\Xi$. Because the realisations at different positions $(i,j)$ are independent under $\mu_\Xi$, the random variables $\zeta_{ij} := w\bigl(v'_{ij},v_{ij}\bigr)
$ are independent. Their distribution depends on $v_{ij}$ and on the probabilities $F_\ell^j(\xi_{ij})$. Denote
\[
p_{ij}^{(0)} = \mathbb{P}_{\Xi}\{\zeta_{ij}=0\},\quad 
p_{ij}^{(1)} = \mathbb{P}_{\Xi}\{\zeta_{ij}=1\},\quad 
p_{ij}^{(2)} = \mathbb{P}_{\Xi}\{\zeta_{ij}=2\}.
\]
These probabilities are easily expressed through $F_{-1}^j,F_0^j,F_1^j$ (see Table~\ref{tab:weights}).

\begin{table}[h]
	\centering
	\begin{tabular}{|c|c|c|c|}
		\hline
		$v'_{ij}$ & $p_{ij}^{(0)}$ & $p_{ij}^{(1)}$ & $p_{ij}^{(2)}$ \\ \hline
		$x_j$      & $F_{1}^j$      & $F_{0}^j$       & $F_{-1}^j$      \\ \hline
		$\overline{x}_j$ & $F_{-1}^j$ & $F_{0}^j$ & $F_{1}^j$ \\ \hline
		$\varepsilon$    & $F_{0}^j$   & $F_{1}^j+F_{-1}^j$ & $0$ \\ \hline
	\end{tabular}
	\caption{Probabilities of the contribution $\zeta_{ij}$ depending on the target literal.}
	\label{tab:weights}
\end{table}

Consequently, the total distance
\[
S := d(v,v') = \sum_{i=1}^n\sum_{j=1}^m \zeta_{ij}
\]
is a sum of independent (but not identically distributed) random variables taking values in $\{0,1,2\}$. 

\begin{thm}
	The probability that $\langle v\rangle$ lies in the ball of radius $\varrho$ (where $\varrho$ is a non‑negative integer) 
	\[\mathbb{P}_{\Xi}\{d(v,v') \le\varrho\} = \sum_{k=0}^{\varrho} c_k,\]
	where $c_k$ are coefficients of the generating function
	\[
	G(t) = \prod_{i=1}^n \prod_{j=1}^m \bigl( p_{ij}^{(0)} + p_{ij}^{(1)} t + p_{ij}^{(2)} t^2 \bigr) = \sum_{k=0}^{2nm} c_k t^k.
	\]
\end{thm}
\begin{proof}
	The probability that $\langle v\rangle$ lies in the ball of radius $\varrho$ is
	\[
	\mathbb{P}_{\Xi}\{d(v,v')\le \varrho\} = \sum_{k=0}^{\varrho} \mathbb{P}_{\Xi}\{S=k\}.
	\]
	Let us define $c_k = \mathbb{P}_{\Xi}\{S=k\}$. The distribution of $S$ can be obtained by means of the probability generating function
	
	\[
	G(t) = \prod_{i=1}^n \prod_{j=1}^m \bigl( p_{ij}^{(0)} + p_{ij}^{(1)} t + p_{ij}^{(2)} t^2 \bigr) = \sum_{k=0}^{2nm} c_k t^k.
	\]
	Hence $\mathbb{P}_{\Xi}\{S\le\varrho\} = \sum_{k=0}^{\varrho} c_k$.
\end{proof}

For moderate $n,m$, the coefficients $c_k$ can be computed by successive convolutions. When $n$ and $m$ are large, a normal approximation may be used:
\[
\mathbb{P}_{\Xi}\{d(v,v')\le \varrho\} \approx \Phi\!\left( \frac{\varrho + 0.5 - \mu_S}{\sigma_S} \right),
\]
where $\Phi$ is the standard normal distribution function and
\[
\mu_S = \sum_{i,j} \mathbb{E}[\zeta_{ij}],\qquad 
\sigma_S^2 = \sum_{i,j} \operatorname{Var}(\zeta_{ij})
\]
with $\mathbb{E}[\zeta_{ij}] = p_{ij}^{(1)} + 2p_{ij}^{(2)}$, $\operatorname{Var}(\zeta_{ij}) = \mathbb{E}[\zeta_{ij}^2] - (\mathbb{E}[\zeta_{ij}])^2$, $\mathbb{E}[\zeta_{ij}^2] = p_{ij}^{(1)} + 4p_{ij}^{(2)}$.

This approximation is accurate when $\mu_S$ is not too close to $0$ or $2nm$, and $\sigma_S$ is large enough.
%

\begin{exa}
	Let us illustrate the computation with a simple case.  
	Take $m = 1$ sensor and $n = 2$ time steps, so the set $\mathcal{V}(1,2)$ consists of venjunctions of the form  
	$\langle v\rangle = v_{21}\angle v_{11}$ (recall that the order is reversed, i.e. the rightmost conjunction corresponds to the earliest time).  
	Choose the target venjunction  
	\[
	\langle v'\rangle = x_1 \angle x_1 \quad (\text{so } v'_{11}=x_1,\; v'_{21}=x_1).
	\]
	
	Assume that for both positions the probability maps are identical and given by  
	\[
	F_1(\xi) = 0.6,\quad F_0(\xi)=0.3,\quad F_{-1}(\xi)=0.1
	\]
	(the value of $\xi$ is irrelevant because we keep it fixed).  
	Thus for each position $(i,1)$ we have
	\[
	p^{(0)} = \mathbb{P}\{\zeta_{i1}=0\} = F_1 = 0.6,\quad
	p^{(1)} = \mathbb{P}\{\zeta_{i1}=1\} = F_0 = 0.3,\quad
	p^{(2)} = \mathbb{P}\{\zeta_{i1}=2\} = F_{-1} = 0.1.
	\]
	
	The random variables $\zeta_{11}$ and $\zeta_{21}$ are independent and identically distributed.  
	Their common probability generating function is  
	\[
	G_{\zeta}(t) = 0.6 + 0.3t + 0.1t^2.
	\]
	
	The total distance $S = \zeta_{11}+\zeta_{21}$ has generating function 
	\begin{multline*}
		G_S(t) = 0.36 + 2\cdot0.6\cdot0.3\,t + \bigl(2\cdot0.6\cdot0.1 + 0.3^2\bigr)t^2 + 2\cdot0.3\cdot0.1\,t^3 + 0.1^2\,t^4 =\\
		= 0.36 + 0.36\,t + 0.21\,t^2 + 0.06\,t^3 + 0.01\,t^4.
	\end{multline*}
	Hence, the probabilities of the possible distances are	
	\[
	\mathbb{P}\{S=0\}=0.36,\;
	\mathbb{P}\{S=1\}=0.36,\;
	\mathbb{P}\{S=2\}=0.21,\;
	\mathbb{P}\{S=3\}=0.06,\;
	\mathbb{P}\{S=4\}=0.01.
	\]	
	From these, we obtain the cumulative probabilities:
	\[
	\begin{array}{c|c}
		\varrho & \mathbb{P}\{d(v,v')\le\varrho\} \\ \hline
		0 & 0.36 \\
		1 & 0.72 \\
		2 & 0.93 \\
		3 & 0.99 \\
		4 & 1.00
	\end{array}
	\]
	Thus, for example, the chance that a randomly drawn venjunction differs from $v'$ by at most $2$ (i.e., agrees in all positions or has a single $\varepsilon$ or two $\varepsilon$'s, etc.) is $93\%$.
	\smallskip
	\noindent\textbf{Normal approximation.}  
	Here $\mu_S = 2(0.3+2\cdot0.1)=2(0.5)=1.0$ and  
	$\sigma_S^2 = 2\bigl[(0.3+4\cdot0.1) - 0.5^2\bigr] = 2(0.7-0.25)=2\cdot0.45=0.9$, so $\sigma_S\approx0.949$.  
	For $\varrho=2$ the approximation gives
	\[
	\Phi\!\left(\frac{2+0.5-1.0}{0.949}\right)=\Phi(1.58)\approx0.943,
	\]
	which is reasonably close to the exact value $0.93$.
\end{exa}

\section{Control Systems Examples\label{app:samples}}
\begin{exa}[A simple control system with uncertainty]
Consider a system of two sensors: atmospheric pressure $p$ and air temperature $\theta$. The pressure sensor is triggered at time $t$ when the pressure $p(t) > p_0$, and the temperature sensor is triggered at time $t$ when the temperature $\theta(t) > \theta_0$. Furthermore, the sensors can erroneously trigger at time $t$ at high air humidity $\varphi(t)>\varphi_0$. We must perform some actions if both $p(t) > p_0$ and $\theta(t) > \theta_0$ for a sufficiently long time.

If the temperature sensor is triggered at time $t$, we will write $x(t)$. If it is not triggered at time $t$, we will write $\overline{x}(t)$. If the pressure sensor is triggered at time $t$, we will write $y(t)$. If it is not triggered at time $t$, we will write $\overline{y}(t)$. Finally, if $\varphi(t)\leq \varphi_0$, $t=\overline{1,T}$, we should perform actions when 
\[Z = x(1)\wedge y(1)\sqcup \dots \sqcup x(T)\wedge y(T),\quad T>T_0>1.\]

As we actually don't know the values of $\varphi(t)$, we can't be sure that $Z$ corresponds to the real situation and can be used for decision-making. So, we can instead write

\[Z(\varphi) = x(1)^{\varphi(1)}\wedge y(1)^{\varphi(1)}\sqcup \dots \sqcup x(T)^{\varphi(T)}\wedge y(T)^{\varphi(T)}.\]

Here $x(t)^{\varphi(t)} = x(t)$, $y(t)^{\varphi(t)} = y(t)$, $\varphi(t)\leq \varphi_0$. If $\varphi(t) > \varphi_0$, $x(t)^{\varphi(t)}$ can take the values $x(t)$, $\overline{x}(t)$ and  $y(t)^{\varphi(t)}$ can take the values $y(t)$, $\overline{y}(t)$.

We want that if $\theta(t) > \theta_0 \wedge \varphi(t)>\varphi_0$, then $x(t)^{\varphi(t)} = x(t)$, but if $\theta(t) \leq \theta_0 \wedge \varphi(t)>\varphi_0$ then $x(t)^{\varphi(t)} = \overline{x}(t)$. Unfortunately, we have no humidity sensor and can't make the aforementioned inferences. Instead, we know some facts about humidity and sensors, namely:

\begin{itemize}[label=$\triangleright$]
	\item Possible values of temperature, pressure, and humidity are inside open line intervals $(\theta_{\min}, \theta_{\max})$, $(p_{\min}, p_{\max})$, $(0,1)$ correspondingly, all parameter values are uniformly distributed.
	\item The temperature sensor sometimes triggers falsely at $\theta_{x}\leq \theta(t)<\theta_0$, $\varphi(t) >\varphi_0$, $\theta_{x}$ is the false triggering threshold.
	\item The pressure sensor sometimes triggers falsely at $p_y\leq p(t)< p_0$, $\varphi(t) >\varphi_0$, $p_{y}$ is the false triggering threshold.
\end{itemize}

Therefore, there are three possible outputs for the temperature sensor $x^{\varphi(t)}$: the correct triggering $x(t)$, the dubious triggering $\varepsilon(t)$, and the correct not triggering $\overline{x}(t)$. 

We can roughly estimate the probabilities of these outputs as follows:
\begin{multline}
	\label{eq:temp}
	\mathbb{P}\{x^{\varphi(t)}\sim x(t)\}=\frac{\theta_{\max}-\theta_0}{\theta_{\max}-\theta_{\min}},\quad \mathbb{P}\{x^{\varphi(t)}\sim \varepsilon(t)\}=\frac{\theta_0-\theta_x}{\theta_{\max}-\theta_{\min}},\\
	\mathbb{P}\{x^{\varphi(t)}\sim \overline{x}(t)\}=\frac{\theta_x-\theta_{\min}}{\theta_{\max}-\theta_{\min}}.
\end{multline}

The pressure sensor $y^{\varphi(t)}$ also has three possible outputs: the correct triggering $y(t)$, the dubious triggering $\varepsilon(t)$, and the correct not triggering $\overline{y}(t)$. We can estimate the probabilities of these outputs as follows:
\begin{multline}
	\label{eq:press}
	\mathbb{P}\{y^{\varphi(t)}\sim y(t)\}=\frac{p_{\max}-p_0}{p_{\max}-p_{\min}},\quad \mathbb{P}\{y^{\varphi(t)}\sim \varepsilon(t)\}=\frac{p_0-p_y}{p_{\max}-p_{\min}},\\
	\mathbb{P}\{y^{\varphi(t)}\sim \overline{y}(t)\}=\frac{p_y-p_{\min}}{p_{\max}-p_{\min}}.
\end{multline}

\begin{figure}[htbp]
	\centering
	\newcommand{\drawtempaxes}{%
		\draw[->, >=stealth, thin, black!80] (-5.2,0) -- (5.3,0) node[right, font=\small] {$\theta - \theta_x$};
		\draw[->, >=stealth, thin, black!80] (0,-0.05) -- (0,1.12);
		\foreach \x in {-4.5,-4,...,4.5} {
			\draw[black!60, thin] (\x,0) -- (\x,0.025);
		}
		\foreach \x in {-4,-2,2,4} {
			\node[below, font=\scriptsize\sffamily, black!80] at (\x,-0.02) {$\x$};
		}
		\foreach \y in {0.2,0.4,0.6,0.8,1.0} {
			\draw[black!60, thin] (0,\y) -- (0.12,\y);
			\node[left, font=\scriptsize\sffamily, black!80] at (-0.05,\y) {\y};
		}
	}
	
	\begin{subfigure}{0.47\textwidth}
		\centering
		\begin{tikzpicture}[x=0.52cm, y=2.6cm]
			\drawtempaxes
			\draw[color=graphblue, very thick, line cap=round] (-5,1) -- (0,1) -- (1,0) -- (5,0);
			\fill[red!60] (1,0) circle (2.2pt);
			\node[above right, font=\scriptsize, xshift=1pt, yshift=1pt] at (1,0) {$\theta_0 - \theta_x$};
		\end{tikzpicture}
		\caption{}
	\end{subfigure}
	\hfill
	\begin{subfigure}{0.47\textwidth}
		\centering
		\begin{tikzpicture}[x=0.52cm, y=2.6cm]
			\drawtempaxes
			\draw[red!60, dashed, thick] (1,0) -- (1,1);
			\draw[color=graphblue, very thick, line cap=round] (-5,0) -- (0,0) -- (1,1);
			\draw[color=graphblue, very thick, line cap=round] (1,0) -- (5,0);
			\fill[red!60] (1,0) circle (2.2pt);
			\node[above right, font=\scriptsize, xshift=1pt, yshift=1pt] at (1,0) {$\theta_0 - \theta_x$};
		\end{tikzpicture}
		\caption{}
	\end{subfigure}
	
	\vspace{0.3cm}
	
	\begin{subfigure}{0.47\textwidth}
		\centering
		\begin{tikzpicture}[x=0.52cm, y=2.6cm]
			\drawtempaxes
			\draw[red!60, dashed, thick] (1,0) -- (1,1);
			\draw[color=graphblue, very thick, line cap=round] (-5,0) -- (1,0);
			\draw[color=graphblue, very thick, line cap=round] (1,1) -- (5,1);
			\fill[red!60] (1,0) circle (2.2pt);
			\node[above right, font=\scriptsize, xshift=1pt, yshift=1pt] at (1,0) {$\theta_0 - \theta_x$};
		\end{tikzpicture}
		\caption{}
	\end{subfigure}
	\caption{The ``temperature'' components $F_{-1}^x$ (a), $F_0^x$ (b), $F_1^x$ (c) of the mapping $\mathbf{F}$}
	\label{fig:fmap0}
\end{figure}

If we can somehow find out the correct temperature and pressure values (for example, through additional sensors without the ability to trigger) it is possible to go further and consider the PDNF $x^{\theta(t)}(t)\wedge y^{p(t)}(t)$. When determining probabilities, it is useful to somehow evaluate the fact that for $\theta(t)$ close to $\theta_x$, $\theta_x<\theta(t)<\theta_0$, the sensor may or may not work correctly. In this case, the mapping $\mathbf{F}$ should be defined as
\begin{gather}\label{eq:FFun01}
	F^x_{-1}(\theta) = \begin{cases}
		1, &\theta<\theta_x,\\
		-(\theta-\theta_x)/(\theta_0-\theta_x)+1, &\theta_x\leq \theta<\theta_0,\\
		0, &\theta \geq \theta_0,
	\end{cases}\\
	F^x_{0}(\theta) = \begin{cases}
		0, &\theta<\theta_x,\\
		(\theta-\theta_x)/(\theta_0-\theta_x), &\theta_x\leq \theta<\theta_0,\\
		0, &\theta \geq \theta_0,
	\end{cases}\\\label{eq:FFun03}
	F^x_{1}(\theta) = \begin{cases}
		0, &\theta<\theta_0,\\
		1, &\theta \geq \theta_0,
	\end{cases}
\end{gather}
for the temperature component $(F_{-1}^x, F_0^x, F_1^x)$ (see Fig. \ref{fig:fmap0}) and, in a similar way, for the pressure component $(F_{-1}^y, F_0^y, F_1^y)$. 

It should be noted that the sequence of the system's states in time moments $t_1,\dots, t_m$ can be written as PDNF $x^{\theta(t_1)}(t_1)\wedge y^{p(t_1)}(t_1)\sqcup\dots \sqcup x^{\theta(t_n)}(t_n)\wedge y^{p(t_n)}(t_n)$. In practice, one may use perceived values $\tilde{\theta}(t)$ and $\tilde{p}(t)$ (obtained from the sensors or other sources) instead of inaccessible true values $\theta(t)$ and $p(t)$ to assess the risk of the sensor being in the dubious state via $F_0^x$ and $F_0^y$. 

Assume that the measured values satisfy $\tilde{\theta}(t) \geq \theta_0$ and $\tilde{p}(t) \geq p_0$, but are close to the thresholds in the sense that $|\tilde{\theta}(t) - \theta_0| < \beta_\theta$ and $|\tilde{p}(t) - p_0| < \beta_p$ for small positive constants $\beta_\theta, \beta_p$. This proximity raises the possibility that the true values $\theta(t), p(t)$ actually lie within the ``dubious intervals'' $[\theta_x, \theta_0)$ and $[p_y, p_0)$. To guard against such measurement errors, we introduce safety margins $\lambda_\theta = \sigma_\theta(\theta_0-\theta_x)$, $\lambda_p =\sigma_p(p_0-p_y)$, where $\sigma_\theta, \sigma_p \in (0,1)$ are heuristically chosen parameters. The action is taken only if for all times $t_1,\dots,t_n$ we have 
\[F_1^x\bigl(\tilde{\theta}(t_i) - \lambda_\theta\bigr) =1,\quad F_1^y\bigl(\tilde{p}(t_i)- \lambda_p\bigr) =1\]
and the sensors are triggered. 

Therefore, we should use the mapping $\mathbf{F}_C$, corrected as described above, for decision-making. Thus, instead of a rather cumbersome verbal description of the sensor system and its decision rules, we obtained a mathematical expression reminiscent of Boolean algebra. 
\end{exa}

Using the method described in the article \cite{Kuznetsov2025PDNF}, we can later reconstruct the PDNF based on actual observations from the sensor system and compare the resulting PDNF with the theoretically predicted one. A discrepancy between the reconstructed and theoretical PDNF would indicate the existence of an unknown systematic error in the sensors.

\begin{exa}[The connection between PDNFs and random walks]
Assume that the temperature $\theta$ and pressure $p$ from the previous example are random walks, for example
\begin{gather*}
	\theta(t) = 20+\frac{1}{4}W\left (\frac{1}{2},\frac{1}{3}\right),\\
	p(t) = 100+\frac{1}{2}W\left (\frac{1}{2},\frac{1}{3}\right),
\end{gather*}
where $W$ is the lattice random walk with the probability of a positive unit step $\frac{1}{2}$, the probability of a negative unit step $\frac{1}{3}$, and the probability of a zero unit step $\frac{1}{6}$. Let the time $t$ be in the set $\mathbf{T}=\{30, \dots, 39\}$, $\theta_x=20$, $\theta_0 = 22$, $p_y = 105$, $p_0=106$. Time is measured in seconds, temperature in Celsius degrees, pressure in kilopascals. We make a very simple numerical simulation of this model in the computer algebra system Wolfram Mathematica.

Introduce the random variable $\mathrm{error}$, which takes the value $0$ if a false triggering occurs in the dubious zone $[\theta_x, \theta_0)$ (meaning that the true physical state should be classified as non-triggered rather than triggered) and the value $1$ otherwise, with each value having probability $\frac{1}{2}$. Let us compare the averaged values of functions (\ref{eq:FFun01})--(\ref{eq:FFun03}), calculated for each $t\in \mathbf{T}$ over $10000$ experiments, with the probabilities of the corresponding ground-truth events under false-alarm suppression:
\begin{gather}\label{eq:TFFun01}
	\chi^x_{-1}(t) = \theta(t)<\theta_x \vee \theta(t) \geq \theta_x \wedge \theta(t) < \theta_0\wedge error = 0,\\
	\chi^x_{0}(t)=\theta(t) \geq \theta_x \wedge \theta(t) < \theta_0 \wedge error = 1,\\\label{eq:TFFun03}
	\chi^x_{1}(t)=\theta(t)\geq \theta_0.
\end{gather}

\begin{figure}[htbp]
	\centering
	\newcommand{\pdnfpt}[2]{%
		\draw[graphblue!40, thin] (#1,0) -- (#1,#2);
		\filldraw[fill=white, draw=graphblue, thick] (#1,#2) circle (1.8pt);
	}
	\newcommand{\rwpt}[2]{%
		\draw[heatdarkorange!50, thin] (#1,0) -- (#1,#2);
		\filldraw[fill=white, draw=heatdarkorange, thick] (#1,#2) ++(0,2.2pt) -- ++(-2pt,-3.5pt) -- ++(4pt,0) -- cycle;
	}
	\newcommand{\drawtfaxes}[1]{%
		\draw[->, >=stealth, thin, black!80] (0,0) -- (10.7,0) node[right, font=\scriptsize\sffamily] {time};
		\draw[->, >=stealth, thin, black!80] (0,0) -- (0,#1) node[above, font=\scriptsize\sffamily] {Probability};
		\foreach \x in {1,...,10} {
			\draw[black!60, thin] (\x,0) -- (\x,0.06);
		}
		\foreach \x in {2,4,6,8,10} {
			\node[below, font=\scriptsize\sffamily, black!80] at (\x,-0.04) {\x};
		}
	}
	
	\begin{subfigure}{0.47\textwidth}
		\centering
		\begin{tikzpicture}[x=0.45cm, y=1.35cm]
			\drawtfaxes{2.25}
			\foreach \y/\lbl in {0.2/0.34, 0.6/0.36, 1.0/0.38, 1.4/0.40, 1.8/0.42} {
				\draw[black!60, thin] (0,\y) -- (0.15,\y);
				\node[left, font=\scriptsize\sffamily, black!80] at (-0.05,\y) {\lbl};
			}
			\foreach \x/\y in {1/1.84, 2/1.62, 3/1.42, 4/1.20, 5/1.00, 6/0.82, 7/0.64, 8/0.46, 9/0.28, 10/0.12} {\pdnfpt{\x}{\y}}
			\foreach \x/\y in {1/1.64, 2/1.46, 3/1.28, 4/1.12, 5/0.94, 6/0.78, 7/0.62, 8/0.46, 9/0.30, 10/0.14} {\rwpt{\x}{\y}}
		\end{tikzpicture}
		\caption{}
	\end{subfigure}
	\hfill
	\begin{subfigure}{0.47\textwidth}
		\centering
		\begin{tikzpicture}[x=0.45cm, y=1.35cm]
			\drawtfaxes{2.25}
			\foreach \y/\lbl in {0.167/0.23, 0.500/0.24, 0.833/0.25, 1.167/0.26, 1.500/0.27, 1.833/0.28} {
				\draw[black!60, thin] (0,\y) -- (0.15,\y);
				\node[left, font=\scriptsize\sffamily, black!80] at (-0.05,\y) {\lbl};
			}
			\foreach \x/\y in {1/1.18, 2/1.17, 3/1.05, 4/0.88, 5/0.72, 6/0.53, 7/0.38, 8/0.33, 9/0.28, 10/0.10} {\pdnfpt{\x}{\y}}
			\foreach \x/\y in {1/1.80, 2/1.62, 3/1.42, 4/1.23, 5/1.03, 6/0.87, 7/0.68, 8/0.50, 9/0.33, 10/0.17} {\rwpt{\x}{\y}}
		\end{tikzpicture}
		\caption{}
	\end{subfigure}
	
	\vspace{0.3cm}
	
	\begin{subfigure}{0.85\textwidth}
		\centering
		\begin{tikzpicture}[x=0.45cm, y=1.35cm]
			\drawtfaxes{2.25}
			\foreach \y/\lbl in {0.267/0.32, 0.533/0.34, 0.800/0.36, 1.067/0.38, 1.333/0.40, 1.600/0.42, 1.867/0.44} {
				\draw[black!60, thin] (0,\y) -- (0.15,\y);
				\node[left, font=\scriptsize\sffamily, black!80] at (-0.05,\y) {\lbl};
			}
			\foreach \x/\y in {1/0.21, 2/0.36, 3/0.55, 4/0.76, 5/0.97, 6/1.15, 7/1.33, 8/1.48, 9/1.61, 10/1.79} {\pdnfpt{\x}{\y}}
			\foreach \x/\y in {1/0.12, 2/0.32, 3/0.52, 4/0.72, 5/0.91, 6/1.08, 7/1.27, 8/1.44, 9/1.61, 10/1.79} {\rwpt{\x}{\y}}
			
			\begin{scope}[xshift=5.6cm, yshift=1.3cm]
				\filldraw[fill=white, draw=graphblue, thick] (0,0.35) circle (1.8pt);
				\node[right, font=\footnotesize\sffamily] at (0.25,0.35) {PDNF Approximation};
				
				\filldraw[fill=white, draw=heatdarkorange, thick] (0,0) ++(0,2.2pt) -- ++(-2pt,-3.5pt) -- ++(4pt,0) -- cycle;
				\node[right, font=\footnotesize\sffamily] at (0.25,0) {Random Walk};
			\end{scope}
		\end{tikzpicture}
		\caption{}
	\end{subfigure}
	\caption{Comparison of averaged values of the mapping components $F_{-1}^x$ (A), $F_0^x$ (B), $F_1^x$ (C) (circles) and the corresponding values $\chi^x_{-1}(t)$, $\chi^x_{0}(t)$, $\chi^x_{1}(t)$ (triangles)}
	\label{fig:tfmap0}
\end{figure}

We obtain (see Fig. \ref{fig:tfmap0}) that averages of (\ref{eq:FFun01})--(\ref{eq:FFun03}) are essentially approximations for probabilities (\ref{eq:TFFun01})--(\ref{eq:TFFun03}). Similar graphs can be obtained for the ``pressure'' component (see Fig. \ref{fig:tfmap1}), where we consider probabilities of the following events
\begin{gather}\label{eq:TFFun01p}
	\chi^y_{-1}(t) = p(t)<p_y \vee p(t) \geq p_y \wedge p(t) < p_0\wedge error = 0,\\
	\chi^y_{0}(t)=p(t) \geq p_y \wedge p(t) < p_0 \wedge error = 1,\\\label{eq:TFFun03p}
	\chi^y_{1}(t)=p(t)\geq p_0.
\end{gather}

We obtain (see Fig.~\ref{fig:tfmap0}) that the averages of (\ref{eq:FFun01})--(\ref{eq:FFun03}) closely track the empirical probabilities (\ref{eq:TFFun01})--(\ref{eq:TFFun03}); the slight systematic offset inside the dubious interval $[\theta_x, \theta_0)$ arises because the PDNF components vary linearly across $[\theta_x, \theta_0)$, whereas the simulated indicator $\mathrm{error} \in \{0,1\}$ assigns a constant conditional probability of $1/2$ throughout that interval.

\begin{figure}[htbp]
	\centering
	\newcommand{\pdnfpt}[2]{%
		\draw[graphblue!40, thin] (#1,0) -- (#1,#2);
		\filldraw[fill=white, draw=graphblue, thick] (#1,#2) circle (1.8pt);
	}
	\newcommand{\rwpt}[2]{%
		\draw[heatdarkorange!50, thin] (#1,0) -- (#1,#2);
		\filldraw[fill=white, draw=heatdarkorange, thick] (#1,#2) ++(0,2.2pt) -- ++(-2pt,-3.5pt) -- ++(4pt,0) -- cycle;
	}
	\newcommand{\drawtfaxes}[1]{%
		\draw[->, >=stealth, thin, black!80] (0,0) -- (10.7,0) node[right, font=\scriptsize\sffamily] {time};
		\draw[->, >=stealth, thin, black!80] (0,0) -- (0,#1) node[above, font=\scriptsize\sffamily] {Probability};
		\foreach \x in {1,...,10} {
			\draw[black!60, thin] (\x,0) -- (\x,0.06);
		}
		\foreach \x in {2,4,6,8,10} {
			\node[below, font=\scriptsize\sffamily, black!80] at (\x,-0.04) {\x};
		}
	}
	
	\begin{subfigure}{0.47\textwidth}
		\centering
		\begin{tikzpicture}[x=0.45cm, y=1.35cm]
			\drawtfaxes{2.25}
			\foreach \y/\lbl in {0.133/0.76, 0.400/0.78, 0.667/0.80, 0.933/0.82, 1.200/0.84, 1.467/0.86, 1.733/0.88} {
				\draw[black!60, thin] (0,\y) -- (0.15,\y);
				\node[left, font=\scriptsize\sffamily, black!80] at (-0.05,\y) {\lbl};
			}
			\foreach \x/\y in {1/1.88, 2/1.75, 3/1.59, 4/1.47, 5/1.31, 6/1.17, 7/1.03, 8/0.87, 9/0.68, 10/0.52} {\pdnfpt{\x}{\y}}
			\foreach \x/\y in {1/1.52, 2/1.36, 3/1.21, 4/1.05, 5/0.89, 6/0.73, 7/0.57, 8/0.41, 9/0.25, 10/0.09} {\rwpt{\x}{\y}}
		\end{tikzpicture}
		\caption{}
	\end{subfigure}
	\hfill
	\begin{subfigure}{0.47\textwidth}
		\centering
		\begin{tikzpicture}[x=0.45cm, y=1.35cm]
			\drawtfaxes{2.25}
			\foreach \y/\lbl in {0.333/0.01, 0.667/0.02, 1.000/0.03, 1.333/0.04, 1.667/0.05} {
				\draw[black!60, thin] (0,\y) -- (0.15,\y);
				\node[left, font=\scriptsize\sffamily, black!80] at (-0.05,\y) {\lbl};
			}
			\foreach \x/\y in {1/0.66, 2/0.67, 3/0.71, 4/0.74, 5/0.73, 6/0.80, 7/0.86, 8/0.85, 9/0.90, 10/0.89} {\pdnfpt{\x}{\y}}
			\foreach \x/\y in {1/1.48, 2/1.54, 3/1.59, 4/1.64, 5/1.69, 6/1.72, 7/1.76, 8/1.79, 9/1.82, 10/1.85} {\rwpt{\x}{\y}}
		\end{tikzpicture}
		\caption{}
	\end{subfigure}
	
	\vspace{0.3cm}
	
	\begin{subfigure}{0.85\textwidth}
		\centering
		\begin{tikzpicture}[x=0.45cm, y=1.35cm]
			\drawtfaxes{2.25}
			\foreach \y/\lbl in {0.50/0.05, 1.00/0.10, 1.50/0.15} {
				\draw[black!60, thin] (0,\y) -- (0.15,\y);
				\node[left, font=\scriptsize\sffamily, black!80] at (-0.05,\y) {\lbl};
			}
			\foreach \x/\y in {1/0.88, 2/0.97, 3/1.08, 4/1.16, 5/1.28, 6/1.37, 7/1.46, 8/1.57, 9/1.70, 10/1.82} {\pdnfpt{\x}{\y}}
			\foreach \x/\y in {1/0.92, 2/1.02, 3/1.12, 4/1.22, 5/1.33, 6/1.43, 7/1.54, 8/1.65, 9/1.76, 10/1.87} {\rwpt{\x}{\y}}
			
			\begin{scope}[xshift=5.6cm, yshift=1.3cm]
				\filldraw[fill=white, draw=graphblue, thick] (0,0.35) circle (1.8pt);
				\node[right, font=\footnotesize\sffamily] at (0.25,0.35) {PDNF Approximation};
				
				\filldraw[fill=white, draw=heatdarkorange, thick] (0,0) ++(0,2.2pt) -- ++(-2pt,-3.5pt) -- ++(4pt,0) -- cycle;
				\node[right, font=\footnotesize\sffamily] at (0.25,0) {Random Walk};
			\end{scope}
		\end{tikzpicture}
		\caption{}
	\end{subfigure}
	\caption{Comparison of averaged values of the mapping component $(F_{-1}^y, F_0^y, F_1^y)$ (circles) and the tuple $(\chi^y_{-1}(t), \chi^y_{0}(t),\chi^y_{1}(t))$ (triangles)}
	\label{fig:tfmap1}
\end{figure}
\end{exa}

\end{document}